\documentclass{article}

\usepackage{amsmath}
\usepackage{amsfonts}
\usepackage{amssymb}
\usepackage{amsthm}
\usepackage{enumerate}
\usepackage{fullpage}
\usepackage{thmtools}

\newtheorem{theorem}{Theorem}[section]
\newtheorem{lemma}[theorem]{Lemma}
\newtheorem{definition}[theorem]{Definition}

\newtheorem{claim}[theorem]{Claim}

\newcommand{\F}{\mathbb{F}}

\newcommand{\C}{\mathbb{C}}
\newcommand{\Tr}{\mathsf{Tr}}
\newcommand{\TR}{\mathbf{TR}}

\newcommand{\Z}{\mathbb{Z}}

\newcommand{\rank}{\mathsf{rank}}

\renewcommand{\deg}{\mathsf{deg}}
\newcommand{\pdeg}{\mathsf{pdeg}}

\newcommand{\poly}{\mathrm{poly}}

\newcommand{\tildeG}{\widetilde{G}}

\newcommand{\Fqbar}{\overline{\F}_q}

\newcommand{\mult}{\mathsf{mult}}

\usepackage{parskip}

\begin{document}
\title{Recovering polynomials over finite fields\\from noisy character values}
\author{Swastik Kopparty\thanks{Department of Mathematics and Department of Computer Science, University of Toronto, Canada. Research supported by an NSERC Discovery Grant.
Email: \texttt{swastik.kopparty@utoronto.ca}}
}
\date{}
\maketitle

\begin{abstract}
Let $g(X)$ be a polynomial over a finite field ${\mathbb F}_q$ with degree $o(q^{1/2})$, and let $\chi$ be the quadratic residue character. We give a polynomial time algorithm to recover $g(X)$ (up to perfect square factors) given the values of $\chi \circ g$ on ${\mathbb F}_q$, with up to a constant fraction of the values having errors. This was previously unknown even for the case of no errors.

We give a similar algorithm for additive characters of polynomials over fields of characteristic $2$. This gives the first polynomial time algorithm for decoding dual-BCH codes of polynomial dimension from a constant fraction of errors.

Our algorithms use ideas from Stepanov's polynomial method proof of the classical Weil bounds on character sums, as well as from the Berlekamp-Welch decoding algorithm for Reed-Solomon codes. A crucial role is played by what we call {\em pseudopolynomials}: high degree polynomials, all of whose derivatives behave like low degree polynomials on ${\mathbb F}_q$.

Both these results can be viewed as algorithmic versions of the Weil bounds for this setting.  
\end{abstract}

\section{Introduction}

This paper is about efficient algorithms for some noisy polynomial interpolation problems over finite fields, where only a small amount of information about each evaluation is available.

Let $q$ be an odd prime power, and let $\chi : \F_q \to \{0, \pm 1\}$ be the quadratic residue character (which takes values $0$ on $0$, $+1$ on nonzero perfect squares, $-1$ on everything else). Let $g(X) \in \F_q[X]$ be an unknown polynomial of degree at most $d \ll q$. We consider the problem of computing $g$ given a noisy version of $\chi \circ g : \F_q \to \{ 0, \pm 1\}$.
Concretely, given $r: \F_q \to \{0,\pm 1\}$, such that for $99\%$ of the $\alpha \in \F_q$ we have $r(\alpha) = \chi(g(\alpha))$, can we efficiently recover $g(X)$?

There is an immediate obstacle to recovering $g$: it may not be uniquely specified by $r$. This arises because $\chi$ cannot detect perfect square factors of $g$. Indeed, if $g_1(X), g_2(X) \in \F_q[X]$ of degree at most $d$ satisfy $g_2(X) = g_1(X) \cdot h^2(X)$ for $h(X) \in \F_q[X]$, then
whenever $h(\alpha) \neq 0$:
$$\chi \circ g_2(\alpha) = \left(\chi \circ g_1(\alpha)\right) \cdot \left(\chi \circ h^2 (\alpha) \right) = \chi \circ g_1(\alpha),$$ and so the Hamming distance $\Delta(\chi \circ g_1, \chi \circ g_2) \leq d \ll q$.

The deep and fundamental Weil bounds for character sums imply that this is essentially the only obstacle to uniqueness. It states that if $f(X)$ is a polynomial of degree at most $d$ that is not of the form $\lambda \cdot h^2(X)$ for $\lambda \in \F_q$ and $h(X) \in \F_q[X]$,
then:
$$ \left|\sum_{\alpha \in \F_q} \chi(f(\alpha)) \right| \leq 2 d \sqrt{q}.$$
(In words, if $d$ is small, a random $\alpha \in \F_q$ has $f(\alpha)$ being a perfect square with probability approximately $1/2$).

Thus if $g_1(X), g_2(X) \in \F_q[X]$ are distinct, monic (the leading coefficient is $1$), squarefree (no irreducible factor appears with multiplicity $\geq 2$), and of degree at most $o(\sqrt{q})$, then $\chi \circ g_1$ and $\chi \circ g_2$ are not only distinct, they are almost orthogonal:
$$  \left|\langle \chi \circ g_1, \chi \circ g_2 \rangle \right| = \left|\sum_{\alpha \in \F_q} \chi(g_1(\alpha)) \cdot \chi(g_2(\alpha)) \right| = \left|\sum_{\alpha \in \F_q} \chi((g_1 \cdot g_2)(\alpha)) \right| \leq O(d \sqrt{q}) = o(q).$$
This implies that as $g$ varies among the monic, squarefree polynomials of degree at most $d = o(\sqrt{q})$,
the vectors $\chi \circ g \in \{0, \pm1\}^q$ have all pairwise Hamming distances $\Delta( \chi \circ g_1 , \chi \circ g_2)$ at least $\left(\frac{1}{2} - \Omega(\frac{d}{\sqrt{q}}) \right) \cdot q$. This defines an error-correcting code, and makes the problem of finding $g$ given a noisy version of $\chi \circ g$ that of unique decoding of this error-correcting code.

We give a polynomial time algorithm for recovering such $g$ from $\chi \circ g$ in the presence of a constant fraction of errors.

\noindent
{\bf Theorem 1 }

{\em Let $\epsilon > 0$.
Let $g(X) \in \F_q[X]$ be monic, squarefree, with $\deg(g) \leq \epsilon \sqrt{q}$.

Suppose $r: \F_q \to \{0, \pm 1 \}$ is such that:
$$ \Delta( r, \chi \circ g) < \left(\frac{1}{8} - \epsilon \right) q.$$

There is a $\poly(q)$ time algorithm, which 
when given $r$, computes $g(X)$.
}

If $g(X)$ is a general monic polynomial of degree $d \leq \epsilon \sqrt{q}$ and $r$ is such that
$\Delta(r, \chi \circ g) < (\frac{1}{8}-\epsilon)q$, then this algorithm will find the unique monic squarefree polynomial $\bar{g}(X)$ of degree at most $d$ such that
$\Delta(\chi \circ g, \chi \circ \bar{g}) \leq d$.

Such an algorithm was previously unknown even for the case of no errors. The previously fastest known algorithm was given by Russell and Shparlinski~\cite{RS04} and ran in $O( d^2 \cdot q^{d+o(1)})$ time. This algorithm made use of the Weil bounds as a black box, and improved upon the trivial $O(q^{d+1})$ time brute force search. 

In~\cite{vDHI}, van Dam, Hallgren and Ip gave efficient {\em quantum} algorithms to solve the $d = 1$ case in time $\poly(\log q)$, which is sublinear in the input size. Russell and Shparlinski~\cite{RS04} showed that the general $d$ case can be solved with $O(d)$ quantum queries, but as far as we understand this does not lead to a fast quantum algorithm.
Further results in this direction were given in \cite{BGKS12,IKSSS18}.

We also prove an analogous result for additive characters over fields $\F_q$ of characteristic $2$. Here the underlying codes are the classical dual-BCH codes, which are the duals of the classical BCH codes~\cite{BC60,H59}.

Let $q = 2^b$, and let $\Tr: \F_q \to \F_2$ be the $\F_2$-linear field trace map given  by $\Tr(\beta) = \sum_{i=0}^{b-1} \beta^{2^i}$. Additive characters $\psi$ of $\F_q$ are closely related to 
$\Tr$ (they are of the form $\psi(\beta) = (-1)^{\Tr(a \cdot \beta)}$), and for notational convenience, we formulate our problem in terms of $\Tr$ instead of $\psi$. 

Let $g(X) \in \F_q[X]$ be an unknown polynomial of degree at most $d$. Suppose for each $\alpha \in \F_q$, we are given a value $r(\alpha)$, such that for $99\%$ of the $\alpha \in \F_q$, $r(\alpha) = \Tr(g(\alpha))$. Can we efficiently recover $g(X)$?

Again, there is an immediate obstacle coming from the fact that 
$r$ (or even $\Tr \circ g$) does not uniquely identify $g(X)$. This arises because
$\Tr$ cannot distinguish between a polynomial and its square. Indeed,
if $g_1(X), g_2(X)$ are such that $g_2(X) = g_1(X) + h(X) + h^2(X)$ for some $h(X) \in \F_q[X]$, then\footnote{This is because $\Tr(\beta+ \beta^2) = 0$ for all $\beta \in \F_q$.} for all $\alpha \in \F_q$:
$$\Tr \circ g_2(\alpha) = \left(\Tr \circ g_1 (\alpha) \right) + \left(\Tr \circ (h + h^2)(\alpha)\right) = \Tr \circ g_1(\alpha),$$
and so $\Tr \circ g_1 = \Tr \circ g_2$.

The Weil bound for additive character sums tell us that this is the only obstacle to uniqueness. It says that if $f(X) \in \F_q[X]$ is of degree at most $d$ and not of the form $\lambda + h(X) + h(X)^2$ for $\lambda \in \F_q$ and $h(X) \in \F_q[X]$, then:
$$ \left| \sum_{\alpha \in \F_q} (-1)^{\Tr(f(\alpha))}  \right|  \leq (d-1)\sqrt{q}.$$
(In words, if $d$ is small, a random $\alpha \in \F_q$ has  $\Tr(f(\alpha))$ being roughly equally likely to be $0$ and $1$).

Thus if $g_1(X), g_2(X) \in \F_q[X]$ are distinct polynomials with only monomials of odd degree, and of degree $d = o(\sqrt{q})$,
then $\Tr \circ g_1$ and $\Tr \circ g_2$ are not only distinct, they are almost as far apart as random vectors:
$$ \left| \sum_{\alpha \in \F_q}  (-1)^{\Tr(g_1(\alpha) - g_2(\alpha))} \right| = \left| \sum_{\alpha \in \F_q}  (-1)^{\Tr((g_1- g_2)(\alpha))} \right| \leq O(d \sqrt{q}) = o(q),$$
and so:
$$ \Delta( \Tr \circ g_1 , \Tr \circ g_2) \in \left(\frac{1}{2} \pm O(\frac{d}{ \sqrt{q}})\right) \cdot q .$$

This gives us the definition of the dual BCH code. For an odd integer $d$, codewords are indexed by polynomials $g(X)$ of degree at most $d$ with only odd degree monomials:
$$ g(X) = a_1 X + a_3 X^3 + \ldots + a_d X^d,$$
and the codewords themselves are the functions $\Tr \circ g : \F_q \to \F_2$.
This is a linear code over $\F_2$, and has length $q$ and $\F_2$-dimension $\frac{d+1}{2} \log_2 q$.
The problem of finding $g$ given a noisy version of $\Tr \circ g$ is thus the problem of unique decoding of this dual BCH code.

We give a polynomial time algorithm for recovering $g$ from $\Tr \circ g$ in the presence of a constant fraction of errors.

\noindent
{\bf Theorem 2 }

{\em Let $\epsilon > 0$.
Let $g(X) \in \F_q[X]$ be a polynomial with only odd degree monomials, and with $\deg(g) \leq \epsilon \sqrt{q}$.

Suppose $r: \F_q \to \F_2$ is such that:
$$ \Delta( r, \Tr \circ g) < \left(\frac{1}{8} - \epsilon \right) q.$$

There is a $\poly(q)$ time algorithm, which 
when given $r$, computes $g(X)$.
}

If $g(X)$ is a general polynomial of degree $d \leq \epsilon \sqrt{q}$
with $0$ constant term, and $r$ is such that
$\Delta(r, \Tr \circ g) < (\frac{1}{8}-\epsilon)q$, then this algorithm will find the unique polynomial $\bar{g}(X)$ of degree at most $d$ with only odd degree monomials such $\Tr \circ g = \Tr \circ \bar{g}$.

There were several previously known results for this problem (namely, the problem of decoding dual BCH codes).
\begin{itemize}
 \item For all $d < O(\sqrt{q})$, if there are no errors at all, then $g(X)$ can be found in time $\poly(q)$. This is because the map $g(X) \mapsto \Tr \circ g$ is $\F_2$-linear; so $g(X)$ can be found from $\Tr \circ g$ by solving a system of linear equations over $\F_2$.
 \item For all $d < O(\sqrt{q})$, if the number of errors is at most $O(\frac{q}{d})$, then there is a $\poly(q)$ time algorithm to recover $g(X)$. This is a consequence of efficient decoding algorithms for Reed-Muller codes (multivariate polynomial codes over $\F_2^b$). The key fact is that the function $\Tr \circ g$ is a $b$-variate $\F_2$-polynomial of degree at most $\lceil \log_2 d \rceil$ when $\F_q$ is viewed as $\F_2^b$ via an $\F_2$-linear isomorphism.
 \item For small $d$, the work of Kaufman-Litsyn~\cite{KL-dBCH}, Kaufman-Sudan~\cite{KS-dBCH} and Kopparty-Saraf~\cite{KS-dBCH-list} gave {\em sublinear time} algorithms for decoding dual BCH codes from a constant fraction of errors. The running time of these algorithms is $O( d^d \cdot \poly(\log q))$. For $d = O(\frac{\log q}{\log\log q})$, they run in time $\poly(q)$, while for slightly smaller $d$ they run in time $q^{o(1)}$.
 
 These algorithms use the fact that there are many short linear dependencies between the coordinates of codewords; namely, that their duals have many low weight codewords. 
 Interestingly, they use decoding algorithms for BCH codes as a subroutine.
\end{itemize}

Theorem 2 gives the first algorithm to decode from a constant fraction of errors in 
$\poly(q)$ time for $d > \omega( \frac{\log q}{\log \log q})$.

Our results can be viewed as algorithmic versions of Weil-type bounds for character sums. To the best of our knowledge, the only known proof that the error-correcting codes underlying Theorem 1 (for $d \geq 3$) and Theorem 2 (for $d \geq \Omega(\log q)$) have distance $\Omega(q)$ go through the Weil bounds. Any algorithm that decodes these codes from $\Omega(q)$ errors must either use a Weil-type bound or prove a Weil-type bounds. Our results fall in the latter category.

We use ideas from Stepanov's polynomial method proof of the Weil bounds, as well as from the Berlekamp-Welch decoding algorithm for Reed-Solomon codes. A crucial role is played by what we call {\em pseudopolynomials}: high degree polynomials, all of whose derivatives behave like low degree polynomials on $\F_q$. 

The number of errors that these algorithms can handle falls short of the unique decoding radius of these codes, which is about $\left( \frac{1}{4} - \epsilon  \right) q$. It would be interesting to get efficient decoding from these many errors.

\subsection{Related Work}

The Weil bounds on character sums were proved by Weil~\cite{weil41,weil48,weil49} 
as a corollary of his proof of the Riemann hypothesis for curves over finite fields. They capture some very basic pseudorandomness phenomena in finite fields in a very strong quantitative form. For this reason, they have found extensive applications in the explicit constructions of pseudorandom objects.
We mention some notable applications in theoretical computer science: the universality of the Paley graphs by Graham and Spencer~\cite{GrahamSpencer}, the construction of epsilon biased sets by Alon, Goldreich, H{\aa}stad and Peralta~\cite{AGHP}, and the construction of affine extractors by Gabizon and Raz~\cite{GR05}.

Another remarkable application is the dual-BCH code, mentioned above. Dual-BCH codes are $\F_2$-linear codes of length $q$ with $q^{(d+1)/2}$ codewords and minimum distance at least $\left(\frac{1}{2} - \frac{d-1}{\sqrt{q}}) \right) \cdot q$.  For small $d$ these are unmatched codes: without the Weil bounds, there is no known proof of existence, for all constant $t$, of binary codes of length $n$ with $n^t$ codewords and minimum distance $(\frac{1}{2} - O_t(\frac{1}{\sqrt{n}})) \cdot n$. (A nonlinear binary code with similar strong parameters can also be made out of the quadratic residue character applied to squarefree polynomials over fields of odd characteristic; one just has to replace the small number of $0$ coordinates with either $+1$ or $-1$, arbitrarily.) This also yields $n^t$ unit vectors in $\mathbb R^n$ with pairwise inner products all at most $O(t/ \sqrt{n})$. This is an even smaller upper bound on the maximum inner product than what random vectors would achieve, $O( \sqrt{(\log n)/n})$.

The original proof of the Weil bounds used machinery from algebraic geometry. Many years later, Stepanov~\cite{Stepanov1969,Stepanov1970,StepanovICM1974} discovered an elementary proof of Weil-type bounds. This used a very clever version of what is now known as the polynomial method (with multiplicities) -- understanding a set of interest by first interpolating a polynomial vanishing on that set (with high multiplicity), and then studying that polynomial.  
Subsequently Schmidt~\cite{SchmidtBook}, Bombieri~\cite{BombieriWeil} and Stohr-Voloch~\cite{SV86} strengthened, simplified and generalized Stepanov's method greatly. See~\cite{SchmidtBook,HKT-Curves-Book,LN-book} for excellent introductions to this and related topics. Our algorithms are heavily influenced by ideas from Stepanov's method.

In recent years, a number of extensions of the Weil bounds for character sums were developed for applications in theoretical computer science and coding theory, see \cite{K11-powering,KL11-weil,CGS16-powering,KTY2,CDGM-traceAG} and ~\cite[Chapter 6]{K10-thesis}.

The problems we consider in this paper are a kind of noisy polynomial interpolation problem. This is an old and very well studied topic in coding theory and theoretical computer science. Our algorithms are also heavily influenced by the polynomial method (with multiplicities) techniques from this domain, especially the ``interpolation + algebra'' framework of the fundamental Berlekamp-Welch and Guruswami-Sudan decoding algorithms for Reed-Solomon codes~\cite{BW85,S96,GS99}. 

The Guruswami-Sudan algorithm solves a polynomial reconstruction problem of a very general type that nicely captures the problem of recovering polynomials from character values. We now describe this problem, known as the list-recovery problem for Reed-Solomon codes. In list-recovery of Reed-Solomon codes, for every $\alpha \in \F_q$ we are given a set $S_\alpha$ of size at most $\ell$, and we want to find all polynomials $f(X) \in \F_q[X]$ of degree at most $d$ with
$$| \{ \alpha \in \F_q \mid f(\alpha) \in S_\alpha \} | \geq 0.99 q.$$ 
The Guruswami-Sudan algorithm~\cite{GS99} efficiently finds all such $f(X)$ in time $\poly(q)$, provided $\ell < (0.99)^2 \frac{q}{d}$ (and in particular, proves that the number of such $f(X)$ is at most $\poly(q)$ ).
When we are given $r: \F_q \to \{0,\pm1\}$  with 
$\Delta(\chi \circ g, r) \leq 0.01q$ for some low degree $g(X)$, the problem of finding $g(X)$ translates into a list-recovery problem as follows.
For each $\alpha \in \F_q$, let $S_\alpha = \chi^{-1}(r(\alpha))$.
Then the set of all $f(X)$ of degree at most $d$ satisfying
$$| \{ \alpha \in \F_q \mid f(\alpha) \in S_\alpha \} | \geq 0.99 q,$$ 
contains $g(X)$ (and in fact all such $f(X)$ are of the form $\lambda \cdot g(X) \cdot h^2(X)$). Thus the problem of recovering $g$ is a special kind of list-recovery problem.

In this setting, $|S_\alpha| \leq \ell = (q-1)/2$, and $d = \omega(1)$, so we are outside the range of where the Guruswami-Sudan algorithm applies. Our result can be viewed as a way of using the extra algebraic structure of the $S_\alpha$ to nevertheless enable some kind of list-recovery.

Our algorithms use higher order derivative information about the polynomials being interpolated. In this sense, they are related in spirit to a recently flourishing line of work on decoding algorithms for multiplicity codes~\cite{RosenbloomTsfasman,Nielsen01,KSY14,GW13,multcodesurvey,multcode-ld,KRSW23,BHKS23,Tamo24,BHKS24,multcodefast,CZ25,Srivastava25,GHKSS-list,AHS26} (Multiplicity codes are error-correcting codes based on evaluations of polynomials and their derivatives).

The dual-BCH code $C$ with $d = \Theta(q^{0.1})$ 
is a binary code that has:
(1) block length $n$, (2) dimension $k = \Omega(n^{0.1} \log n)$,
(3) distance $\left(\frac{1}{2} - o(1) \right)n$, and
(4) dual distance (i.e., distance of $C^\perp$) $\Omega(n^{0.1})$.
Theorem 2 shows that it also has: (5) efficiently decoding from $\Omega(n)$ errors.

Explicit binary codes with distance $\frac{n}{2} - o(1)$, dual distance $n^{\Omega(1)}$,
and efficiently decodable from $\Omega(n)$ fraction errors
were not known until quite recently. The first such codes were constructed by Kopparty, Shaltiel, and Silbak~\cite{KSS19}.
These codes, known as {\em raw Reed-Solomon codes}, are closely related to dual-BCH 
codes; their codewords are formed by taking juxtapositions of related dual-BCH codewords, and their minimum distance is proved by appealing to the same Weil bounds. The efficient decoding algorithm for these codes, which can actually list-decode from $(\frac{1}{2}-o(1))n$ errors, is based on a suitable instantiation of the Guruswami-Sudan algorithm mentioned above.

If we relax the distance requirement on the code to simply be $\Omega(n)$, then one can even achieve this with dual distance $\Omega(n)$ using algebraic-geometric codes; this is an important result of Ashikhmin, Litsyn and Tsfasman~\cite{ALT} discovered in the context of quantum codes (see also Guruswami's appendix to~\cite{guruswami-shpilka-selfdual}).

\subsection{Techniques}

Let $q$ be odd.
Let $g(X)$ be a monic squarefree polynomial of degree $d = q^{0.1}$.
We first describe the algorithm to recover $g$ when we are given $\chi \circ g$ with no errors.

Let $\tildeG(X) = (g(X))^{(q-1)/2}$. 
Then\footnote{We now view the character $\chi$ as $\F_q$-valued.} $\chi \circ g (\alpha) = \tildeG(\alpha)$ for each $\alpha \in \F_q$.
Our input thus gives us $q$ evaluations of the polynomial $\tildeG(X)$. 
However $\widetilde{G}(X)$ has much higher degree $\approx d\frac{q}{2}$, and this is not enough to interpolate $\tildeG$ from degree considerations alone.

To motivate the algorithm, we begin by noting something striking about the derivative\footnote{This is just the standard derivative of polynomials, defined by $(X^i)' = i X^{i-1}$, and extended by linearity.} of $\tildeG(X)$. 
We calculate:
$$ \tildeG'(X) = \frac{q-1}{2} \cdot g(X)^{(q-3)/2} \cdot g'(X) = -\frac{1}{2}\cdot \frac{g'(X)}{g(X)} \cdot \tildeG(X).$$
Thus for all $\alpha \in \F_q$:
$$ \tildeG'(\alpha) = - \frac{1}{2}\cdot \frac{g'(\alpha)}{g(\alpha)} \cdot (\chi \circ g(\alpha)).$$
Let $S_0, S_1, S_{-1}$ be the subsets of $\F_q$ where $\tildeG$ takes values $0,1,-1$ respectively. We know that $S_0$ is tiny, and that $S_1, S_{-1}$ are both of size about $\frac{q}{2}$ (by the Weil bounds). 
The above calculation shows that there is a very low degree rational function
$h(X) = -\frac{1}{2} g'(X)/g(X)$, such that the huge degree polynomial $\tilde{G}'$ agrees with $h$ on $S_1$ and agrees with $-h$ on $S_{-1}$. This is a very special kind of behavior: not only is a huge degree polynomial becoming a low degree rational function on the large sets $S_1$ and $S_{-1}$, the two rational functions themselves are closely related. This phenomenon continues for the higher derivatives too. Our algorithm will be based on this.

First some small adjustments.
To avoid denominators, we pick an integer $M= O(d)$ and define $G(X) = (g(X))^{(q-1)/2 + M}$. Instead of standard derivatives, we will use Hasse derivatives, which work better over finite fields. Let $G^{[\ell]}(X)$ denote the $\ell$-th (Hasse) derivative of $G(X)$. 
Then we have the following fact: for each $\ell$ with $0 \leq \ell < M$, there is a polynomial $V_{\ell}(X)$ of degree at most $O(dM) \ll q$ such that:
\begin{align}
 \label{eq:GgV}
 G^{[\ell]}(\alpha) =  \left( \chi \circ g (\alpha) \right) \cdot V_\ell(\alpha)
\end{align}
for all $\alpha \in \F_q$.

This gives the first step of the algorithm: by solving a system of linear equations, we search for a polynomial $F(X)$ of degree at most $d \cdot ( (q-1)/2 + M)$, and polynomials $U_0(X), U_1(X), \ldots, U_{M-1}(X)$ of degree at most $O(dM)$, such that for all $\alpha \in \F_q$, $F^{[\ell]}(\alpha) = (\chi \circ g(\alpha)) \cdot U_{\ell}(\alpha)$.

We know these equations have a solution, since $G$ and the $V_\ell$ satisfy them.
But whatever solution $F$ and the $U_\ell$ the algorithm finds may not equal $G$ and the $V_\ell$. This could really happen: if $h_1(X)$ and $h_2(X)$ are very low degree polynomials, and $f(X)$ is of the form $ g(X) \cdot h_2(X)^2$, then taking $F(X)$ to be $h_1(X) \cdot (f(X))^{(q-1)/2 + M}$ also yields a valid solution to these equations.

The next step of the algorithm recovers $g(X)$ from $F(X)$ by factoring $F(X)$ and inspecting the factor multiplicities of $F(X)$. Let
$g(X) = \prod_{i} g_i(X)$ be the factorization of $g(X)$ into distinct irreducible factors. We show that the set of $g_i(X)$ is exactly the set of all irreducible factors of $F(X)$ whose factor multiplicity, after reducing mod $q$, is approximately $q/2$. This is clearly true if $F(X) = G(X)$. That it holds for any $F(X)$ found by the first step of the algorithm is the main step of the analysis.

The proof of this is quite indirect. First, using a polynomial interpolation argument, we show that there are special polynomials $A(X), B(X) \in \F_q[X]$ such that:
\begin{align}\label{introAFBG}
A(X) \cdot F(X) = B(X) \cdot G(X).
\end{align}
These special polynomials are what we call {\em pseudopolynomials}.
They are of the form $C(X) = \sum_{i} C_i(X) (X^q - X)^i$, with $C_i(X)$ having
low degree $k$ (much lower than $q$). The choice of this class of polynomials, as well as the idea of the interpolation argument, come from Stepanov's method.
Pseudopolynomials have two key properties that we use:
\begin{itemize}
 \item For each $\ell$, there is a polynomial $C_{\langle \ell \rangle}(X)$ of degree at most $k$ such that the $\ell$-th derivative $C^{[\ell]}$ and $C_{\langle \ell \rangle}$ agree on all points of $\F_q$. This is used in the interpolation argument, to get the identity~\eqref{introAFBG}.
 \item Every irreducible factor of $C(X)$ has factor multiplicity approximately a multiple of $q$. This ensures, using \eqref{introAFBG}, that the only factors of $F$ with multiplicity approximately $q/2$ mod $q$ are the factors of $G$ with multiplicity approximate $q/2$ mod $q$, as we want.
\end{itemize}
This completes the sketch of the proof in the error-free case.

To handle errors, we use the idea of the Berlekamp-Welch decoding algorithm for Reed-Solomon codes. However, instead of using an error-locator polynomial to zero out the errors, we have to use a high multiplicity error-locator pseudopolynomial $E(X)$. Then ensures that the derivatives of $E(X) \cdot G(X)$ also have a property 
similar to the property of $G(X)$ given in Equation~\eqref{eq:GgV}.

For the additive character case, the algorithm becomes slightly more complicated. 
To recover $g$ given a noisy version $r$ of $\Tr \circ g$, we set up a system of linear equations capturing some special low-degree property of the derivatives of the polynomial $G(X) = \sum_{i=0}^{b-1} g^{2^i}(X)$. Specifically, we search for nonzero polynomials $F(X)$ and $E(X)$, where $E(X)$ is a pseudopolynomial, such that there exist low degree polynomials $U_0(X), \ldots, U_{M-1}(X)$ with:
$$ F^{[\ell]}(\alpha) = E^{[\ell]}(\alpha) \cdot r(\alpha) + U_\ell(\alpha)$$
for all $\alpha \in \F_q$. If $E(X)$ is a high multiplicity error-locator pseudopolynomial and $F(X) = E(X) \cdot G(X)$, then the above equations are satisfied. However the solution found by the algorithm may not be this one.

From this solution to this system of equations, we try to recover $g(X)$; however, unlike in the case of the quadratic residue character, where the relevant problem was factoring polynomials and there was an off-the-shelf algorithm known, we have to work a little harder. Here we only manage to use the solution of the system of equations to extract the largest degree coefficient of $g(X)$. We can then peel it off and repeat, and thus get the whole of $g(X)$. The analysis again uses an interpolation argument involving pseudopolynomials.

Our presentation of the algorithms and their analyses in Sections~\ref{sec:chi} and~\ref{sec:tr} is done from scratch and does not explicitly use the language of pseudopolynomials. Later, in Section~\ref{sec:pseudopoly}, we formally define pseudopolynomials and the pseudodegree, state some basic properties, and explain some ideas from previous sections in a cleaner context. We also formulate Stepanov's proof of the Weil bounds in this language.

In the appendix, we briefly describe how these algorithms generalize to the $m$-th power residue character and additive characters over finite fields of small characteristic.

\subsection{Questions}

Below we highlight some interesting questions.

\begin{enumerate}
\item Is there an efficient algorithm to decode these codes from a $(\frac{1}{4} - \epsilon )$-fraction of errors?
This is the unique decoding radius for these codes. 

\item More ambitiously, we could ask whether there is an efficient algorithm for list-decoding these codes from a $\left(\frac{1}{2} - O(\sqrt \epsilon) \right)$-fraction of errors. This is the (binary) Johnson radius for these codes, so we know that there are only $O_\epsilon(1)$ codewords within this radius of any given function $r$.

In some ways, our result is analogous to the Berlekamp-Welch algorithm decoding algorithm for Reed-Solomon codes. Is there an analogue of the Guruswami-Sudan algorithm for this context?

\item Both the decoding problems that we give algorithms for are instances of the following more general problem.

Let $G(X,Y) \in \F_q[X,Y]$ be an unknown polynomial of degree $d$. Suppose we are given as input the projection $S = \Pi_X(G)$, where:
$$\Pi_X(G) = \{ x \in  \F_q \mid \exists y \in \F_q \mbox{ with } G(x,y) = 0\}.$$
Can we find, in time $\poly(q)$, some polynomial $F(X,Y)$ of degree at most $d$ such that $|\Pi_X(F) \triangle S| \leq O_d(1)$?

Such sets are {\em definable} sets with complexity $O_d(1)$ in the sense of  Chatzidakis, van den Dries and Macintyre~\cite{CDM} and Tao~\cite{T12-expanding,T13-alg-reg}.
The results of \cite{CDM}, which are based on more general forms of the Weil bounds, imply that the symmetric difference of any two such definable sets is of size either $O_d(1)$ or $\Omega_d(q)$. This suggests that one could even hope to solve this given a noisy version of $\Pi_X(G)$.

\item Let $q$ be a prime, and let $H \subset \F_q$ be the set
$\{0,1, \ldots, (q-1)/2\}$ and let $1_H: \F_q \to \{0,1\}$ be its indicator function. 
Suppose $g(X)$ is an unknown polynomial in $\F_q[X]$ of degree at most $d$, and we are given all evaluations of $(1_{H}) \circ g: \F_q \to \{0,1\}$. 

The additive character sum Weil bounds over $\F_q$ imply that if $d < O(q^{\frac{1}{2} - \epsilon})$, then $g(X)$ is uniquely determined, up to the constant term, by the function $1_H \circ g$.

Can we find $g(X)$ (up to the constant term) in time $\poly(q)$? 

To the best of our knowledge, the only elementary proofs of the uniqueness statement above are very indirect, involving the Stepanov method in multiple extension fields $\F_{q^c}$. It will be very interesting to find a more direct proof, and adapt it to the algorithmic problem above.

\end{enumerate}

\section*{Acknowledgements}

I would like to thank Ronen Shaltiel and Divesh Aggarwal, who independently asked me about efficiently decoding dual-BCH codes.

I am very grateful to Divesh Aggarwal, Mitali Bafna, Eli Ben-Sasson, Ariel Gabizon, Mrinal Kumar, Shubhangi Saraf, Ronen Shaltiel, Madhu Sudan, Amnon Ta-Shma, Santoshini Velusamy, Avi Wigderson and Kedem Yakirevitch for valuable discussions on these topics over the years. I would especially like to thank Amnon Ta-Shma for many illuminating discussions about the Stepanov method.

\section{Preliminaries}

$\F_q^*$ denotes $\F_q \setminus \{0\}$.

For two functions $r_1, r_2: \F_q \to \F_q$, 
we define their Hamming distance to be the number of evaluation points where they disagree:
$$ \Delta(r_1, r_2) = |\{\alpha \in \F_q \mid  r_1(\alpha) \neq r_2(\alpha)\}|.$$

We use $\poly(n)$ to denote an expression bounded by $A \cdot n^B$ for absolute constants $A,B$.

\subsection{Polynomials, Hasse Derivatives and Multiplicities}

Let $\F$ be a field.

We use $\deg(A)$ to denote the degree of a polynomial $A(X) \in \F[X]$. For polynomials $A(X), B(X) \in \F[X]$, we will use $A(X) \mod B(X)$ to denote the remainder
when $A(X)$ is divided by $B(X)$.

We now define Hasse derivatives, a notion of higher-order derivative that is suited for polynomials over finite fields.
See~\cite{HKT-Curves-Book,DKSS} for more on this.

\begin{definition}[Hasse Derivatives]
For a polynomial $A(X) \in \F[X]$,
we define the $\ell$-th Hasse derivative of $A(X)$, denoted $A^{[\ell]}(X)$, to be the coefficient $A_\ell(X) \in \F[X]$ of $Z^\ell$ in the expansion:
$$ A(X+Z) = A_0(X) + A_1(X) Z + \ldots + A_{\ell}(X) Z^\ell + \ldots .$$

\end{definition}

\begin{definition}[Multiplicity]
 For a nonzero polynomial $A(X) \in \F[X]$ and $\alpha \in \F$, we define the multiplicity of vanishing of $A(X)$, denoted $\mult(A, \alpha)$, to be the largest $r \in \mathbb N$ such that $(X-\alpha)^r$ divides $A(X)$.
 
 Equivalently, it is the largest $r$ such that 
 for all $\ell$ with $0 \leq \ell < r$,  $$ A^{[\ell]}(\alpha) = 0,$$
 and $ A^{[r]}(\alpha) \neq 0$.
\end{definition}

\begin{lemma}[Hasse derivative rules]
 Let $A(X), B(X) \in \F[X]$, and $\alpha, \beta \in \F$. 
 Then for all $\ell \geq 0$:
 \begin{itemize}
  \item
  $$(\alpha A+ \beta B)^{[\ell]}(X) = \alpha A^{[\ell]}(X) + \beta B^{[\ell]}(X).$$
  \item $$ (A \cdot B)^{[\ell]}(X) = \sum_{\substack{\ell_1+ \ell_2 = \ell \\ \ell_1, \ell_2 \geq 0}} A^{[\ell_1]}(X) \cdot B^{[\ell_2]}(X).$$
 \end{itemize}
\end{lemma}

Now we specialize to the finite field $\F_q$.
Here the polynomial
$$\Lambda(X) = X^q - X$$
plays a special role.

We recall some basic derivative calculations.

For $A(X) = X^i$:
$$ A^{[\ell]}(X) = {i \choose \ell} X^{i-\ell}.$$
For $\Lambda(X)$:
$$ \Lambda^{[\ell]}(X) = \begin{cases}
                    \Lambda(X) & \ell = 0\\
                    -1 & \ell = 1\\
                    1 & \ell = q\\
                    0 & \mbox{ otherwise}
                   \end{cases}.$$

\begin{lemma}
\label{lem:hasse-xq}
Suppose $U(X), V(X) \in \F_q[X]$, $\deg(V)<q$, $i \geq 0$, and:
$$U(X) = V(X) \cdot \Lambda^i(X),$$

Then for $0 \leq \ell < q$:
$$U^{[\ell]}(X) \equiv  (-1)^{i} \cdot V^{[\ell-i]}(X)\mod \Lambda(X),$$

(where we adopt the convention that $V^{[j]}(X) = 0$ if $j < 0$.)
\end{lemma}
\begin{proof}
 By the product rule for Hasse derivatives:
 \begin{align*}
  U^{[\ell]}(X) &= \sum_{\substack{\ell_0, \ell_1, \ldots, \ell_i \\ \sum \ell_j = \ell}} V^{[\ell_0]}(X) \cdot \prod_{j=1}^{i} \Lambda^{[\ell_j]}(X).\\
 \end{align*}
 
 If $\ell < i$, then in every term in the sum one of $\ell_1, \ldots, \ell_i$ must equal $0$, and thus all the terms are multiples of $\Lambda(X)$.
 So $U^{[\ell]}(X)$ is a multiple of $\Lambda(X)$, as claimed.
 
 If $\ell \geq i$, then we will see that all but at most one of these terms is a multiple of $\Lambda(X)$.
 
 Since $\ell < q$, every term has all the $\ell_j < q$.
 
 By our computation of derivatives of $\Lambda$, 
 the only terms that are nonzero must have 
 $\ell_1, \ldots, \ell_i$ all equal to either $0$ or $1$.
 
 If at least one of $\ell_1, \ldots, \ell_i$ equals $0$, then $\Lambda(X)$ is a factor of the term.
 
 The only term remaining is $\ell_1 = \ell_2 = \ldots = \ell_i = 1$, and $\ell_0 = \ell - i$, and this term equals:
 $$ V^{[\ell-i]}(X) \cdot (-1)^{i}.$$
 
 This completes the proof of the lemma.
\end{proof}

\subsection{The quadratic residue character over fields of odd characteristic}

Let $q$ be an odd prime power.

Let $\chi : \F_q \to \{0, \pm 1\}$ be the quadratic residue character:
$$ \chi(\beta) = \begin{cases} +1 & \beta \in \F_q^* \mbox{ is the square of an element of $\F_q$ } \\
                  -1 & \beta \mbox{ is not the square of an element of $\F_q$ }\\
                  0 & \beta = 0.
                 \end{cases}$$
                 Usually the range is viewed as a subset of $\C$, but we will view it as a subset of $\F_q$. Then we have the formula:
$$ \chi(\beta) = \beta^{(q-1)/2}.$$

Let $g(X) \in \F_q[X]$ with $g$ of degree $d$. 
We will use $\chi \circ g$ to denote the {\em function}
from $\F_q$ to $\{0, \pm 1\} \subseteq \F_q$ whose values are given by:
$$ \chi \circ g (\alpha) = \chi(g(\alpha)).$$
(We will never use $\chi \circ g$ to refer to the polynomial $g^{(q-1)/2}(X)$.)

For an arbitrary nonzero $g(X)$ of degree $d$, we can write it as:
$$g(X) = \lambda \cdot \bar{g}(X) \cdot h(X)^2$$
where $\lambda \in \F_q^*$, $\bar{g}(X), h(X) \in \F_q[X]$, satisfy:
\begin{itemize}
 \item $\bar{g}(X)$ is monic and squarefree,
 \item $h(X)$ is monic.
\end{itemize}
$\bar{g}(X)$ is the ``squarefree core'' of $g(X)$, and essentially determines $\chi \circ g$.

Indeed, for all $\alpha$ where $h(\alpha) \neq 0$,
$$ \chi \circ g (\alpha) = \chi(\lambda) \cdot \chi \circ \bar{g} ( \alpha).$$

Thus, either $ \Delta( \chi \circ g, \chi \circ \bar{g}) \leq d$ or $ \Delta( \chi \circ g, -\chi \circ \bar{g}) \leq d$.

\subsection{Additive characters in fields characteristic $2$ }

Let $q$ be a power of $2$, say $q = 2^b$. 

Let $\Tr: \F_q \to \F_2$ be the field trace map:
$$ \Tr(\beta) = \sum_{i=0}^{b-1} \beta^{2^i}.$$

Let $\TR(X) \in \F_q[X]$ be the polynomial:
$$ \TR(X) = \sum_{i=0}^{b-1} X^{2^i}.$$

For a polynomial $g(X) \in \F_q[X]$, we will use $\Tr \circ g$ to denote the {\em function} from $\F_q$ to $\F_2$ whose values are given by:
$$ \Tr \circ g (\alpha) = \Tr(g(\alpha)).$$
(We will never use $\Tr \circ g$ to refer to the polynomial 
$ \TR(g(X)) = \sum_{i=0}^{b-1} (g(X))^{2^i}.$)

For an arbitrary nonzero $g(X)$ of degree $d$, we can write it as:
$$ g(X) = \lambda + \bar{g}(X) + h(X) + h^2(X),$$
where  $\lambda \in \F_q$ and $\bar{g}(X), h(X) \in \F_q[X]$ satisfy:
\begin{itemize}
 \item $\bar{g}(X)$ has only odd degree monomials,
 \item $h(X)$ has $0$ constant term.
\end{itemize}
This is proved by repeatedly taking the highest even degree monomial $a_{2i} X^{2i}$ and writing it as the sum of $(a_{2i}^{1/2} X^i)$ and $(a_{2i}^{1/2} X^i + a_{2i} X^{2i})$ -- the former having lower degree, and the latter being of the form $h(X) + h^2(X)$.

$\bar{g}(X)$ is the ``odd degree core'' of $g(X)$ and essentially determines $\Tr 
\circ g$.

Indeed, for all $\alpha \in \F_q$, 
$$ \Tr \circ g (\alpha) = \Tr ( \lambda) + \Tr \circ \bar{g}(\alpha).$$
Thus either $\Tr \circ g = \Tr \circ \bar{g}$ or $\Tr \circ g = 1 + \Tr \circ \bar{g}$.

\newpage
\section{$\chi \circ g$}
\label{sec:chi}

Let $q$ be an odd prime power.

Let $g(X) \in \F_q[X]$ be a monic squarefree polynomial.
In this section, we give an efficient algorithm to recover $g(X)$ given access to a noisy version $r$ of $\chi \circ g$.

{\noindent \bf Algorithm A:}\\
{\noindent Parameters: degree $d \leq O(\epsilon \sqrt{q})$, error-bound $e \leq \left( \frac{1}{8} - \epsilon \right) q$.}\\
{\noindent Input $r: \F_q \to \{0, \pm 1\}$ }
\begin{enumerate}
 \item Set
 \begin{itemize}
  \item $M = \frac{16}{\epsilon}d$
  \item $c = \frac{M}{2}$
  \item $h = 2 \cdot e$.
  \item $D = d \cdot  \left( (q-1)/2 + M \right) + cq = (1 + O(\epsilon)) \cdot \frac{1}{2} \cdot M \cdot q$,
  \item $u = h + dM = 2e + O(\frac{ d^2}{\epsilon} )$.
 \end{itemize}
 
 \item\label{linsyst} Solve an $\F_q$ system of linear equations to find polynomials
 $F(X), U_0(X), \ldots, U_{M-1}(X) \in \F_q[X]$, not all zero,
 where:
 \begin{itemize}
  \item $\deg(F) \leq D $
 \item for each $\ell$, $\deg(U_\ell) \leq u$
 \item For all $\alpha \in \F_q$, $0 \leq \ell < M$:
 \begin{align}
 \label{eq:linsyst}
  F^{[\ell]}(\alpha)  = r(\alpha) \cdot U_\ell(\alpha).
 \end{align}
  
 \end{itemize}
 \item Factor $F(X)$ into irreducible factors:
 $$ F(X) = \lambda \prod_j  H_j^{\mu_j}(X),$$
 where the $H_j(X) \in \F_q[X]$ are distinct and monic, and $\lambda \in \F_q^*$.
 \item Define
 \label{getJstep}
 $$J = \left\{ j \mbox{ such that }  \mu_j \in \left[\frac{3}{8}q, \frac{7}{8}q\right] \mod q  \right\}.$$
 \item Set
 \label{getfstep}
 $$f(X) = \prod_{j \in J} H_j(X).$$
 \item Return $f(X)$.
\end{enumerate}

It is clear from the description of the algorithm (and the fact that factoring univariate polynomials of degree $D$ over $\F_q$ can be done in time $\poly(q,D)$~\cite{B67-factor}), this algorithm can be implemented to run in time $\poly(q)$.

We now show correctness.

\subsection{Correctness of the Algorithm}

\begin{theorem}
\label{thm:A}
 Let $\epsilon > 0$, and suppose:
 \begin{itemize}
  \item $d \leq \frac{\epsilon}{16} \sqrt{q}$
  \item $e \leq \left( \frac{1}{8} - \epsilon \right) q$.
 \end{itemize}

 Suppose $g(X) \in \F_q[X]$ is a monic squarefree polynomial with degree $\leq d$ such that $\Delta(\chi \circ g, r) \leq e$.
 
 Then Algorithm A returns $g(X)$.
\end{theorem}
\begin{proof}
 The proof has three steps.
 
 Let $G(X) = (g(X))^{(q-1)/2 + M}$.
 
 First, using the hypothesis that $r$ is close to $\chi \circ g$, we show that 
 the system of linear equations in Step~\ref{linsyst}
has a nonzero solution.
Specifically, we carefully design polynomials $E(X)$ (which is a kind of
high multiplicity error-locator pseudopolynomial) and $V_0(X), \ldots, V_{M-1}(X)$
so that $F(X) = E(X) \cdot G(X)$ and $U_i(X) = V_i(X)$ satisfies
the linear equations.
Thus the algorithm will find a nonzero solution in that step.

The algorithm may not find the above mentioned ``intended'' solution.
In the second step of the proof, we show that whatever solution  $(F, U_0, \ldots, U_{M-1})$ it does find, $F(X)$ has some nontrivial relation to $G(X)$.

Finally, we use this nontrivial relation to show that the multiplicities of irreducible factors of $F$ are closely related to those of $G$, and hence of $g$, and deduce the result from that.
 
We now proceed with the details.

\subsubsection*{Step 1: The system of linear equations has a nonzero solution.}

The key observation is that if the received word $r: \F_q \to \{0 , \pm 1\}$
had been exactly $\Tr \circ g$, then $F = G$ would have given a valid solution to the system of equations. Instead, our received word $r$ is merely close to $\chi \circ g$. So, following the idea of the Berlekamp-Welch decoding algorithm for Reed-Solomon codes, we will zero out $G$ at the error locations to get a valid solution.

Let $S = \{\alpha \in \F_q \mid \chi \circ g(\alpha) \neq r(\alpha)\}$ be the error set. Let $Z_S(X)$ be the error locator polynomial:
$$ Z_S(X) = \prod_{\alpha \in S} (X-\alpha).$$

Recall that $c = M/2$ and $h = \frac{eM}{c} = 2e$.
We now take $E(X)$ to be a nonzero multiple of $Z_S^M(X)$ of the form:
$$ E(X) = \sum_{i=0}^{c-1} E_i(X) \Lambda^{i}(X),$$
with $\deg(E_i) \leq h$. (Recall that $\Lambda(X) = X^q - X$). 
Such an $E(X)$ exists, since vanishing mod $Z_S^M(X)$ imposes $eM$ constraints on the $c\cdot(h+1)$-dimensional $\F_q$-linear space $\{(E_0(X), \ldots, E_{c-1}(X)) \in \F_q[X]^{c} \mid \deg(E_i) \leq h \}$,
and $c \cdot (h+1) > ch = eM$.

This gives us two nice properties:
\begin{itemize}
 \item $E^{[\ell]}(\alpha) = 0$ for all $\alpha \in S$ and $\ell$ with $0 \leq \ell < M$. This is because $E$ is a multiple of $Z_S^M(X)$, and thus of $(X-\alpha)^M$.
 \item For each $\ell$ with $0 \leq \ell < M$, there is a polynomial $E_{\langle \ell\rangle}(X) \in \F_q[X]$ of degree at most $h = 2e$ such that for all $\alpha \in \F_q$:
 $$ E_{\langle \ell \rangle}(\alpha) = E^{[\ell]}(\alpha).$$
 Indeed, by Lemma~\ref{lem:hasse-xq}:
 $$E^{[\ell]}(X) = \sum_{i=0}^{c-1} (-1)^{i} E_i^{[\ell-i]}(X)  \mod \Lambda(X),$$
 and thus  if we define:
 $$ E_{\langle \ell \rangle}(X) = \sum_{i=0}^{c-1} (-1)^i E_i^{[\ell-i]}(X),$$
 we have for all $\alpha \in \F_q$,:
 \begin{align*}
 E^{[\ell]}(\alpha) = \sum_{i=0}^{c-1} (-1)^i E_i^{[\ell-i]}(\alpha) = E_{\langle \ell \rangle}(\alpha).
 \end{align*}
\end{itemize}

We will take our solution $F(X)$ to be:
$$ F(X) = E(X) \cdot G(X),$$ 
(recall that $G(X) = g^{(q-1)/2 + M}(X)$).

Before we compute its derivatives, we need a quick lemma.
\begin{lemma}
Let $g(X) \in \F_q[X]$ be a polynomial of degree at most $d$.
Then the $\ell$-th derivative of $g^j(X)$ is of the form:
 $$(g^j)^{[\ell]}(X) = H_{g, j,\ell}(X) \cdot g(X)^{j-\ell},$$
 where $H_{g,j,\ell}(X) \in \F_q[X]$ has degree at most $d \ell - \ell$.
\end{lemma}
\begin{proof}
 Using the product rule for Hasse derivatives.
\end{proof}

Set $w = (q-1)/2 + M$. Then by the above lemma:
$$G^{[\ell]}(X) = H_{g,w,\ell}(X) \cdot g^{w-\ell}(X) = H_{g,w,\ell}(X)  \cdot g^{M-\ell}(X)\cdot g^{(q-1)/2}(X),$$
and so for all $\alpha \in \F_q$:
\begin{align}
G^{[\ell]}(\alpha) =    H_{g,w,\ell}(\alpha) \cdot g^{M-\ell}(\alpha) \cdot g^{(q-1)/2}(\alpha) =  \left(  H_{g,w,\ell}(\alpha) \cdot g^{M-\ell}(\alpha) \right) \cdot \chi \circ g(\alpha).
\end{align}

Now we compute the derivatives of $E \cdot G$ at $\alpha \in \F_q$:
\begin{align}
\nonumber
(E\cdot G)^{[\ell]}(\alpha)  &= \sum_{\ell_1 + \ell_2 = \ell}  E^{[\ell_1]}(\alpha) \cdot G^{[\ell_2]}(\alpha)\\
\nonumber
  &=  \sum_{\ell_1 + \ell_2 = \ell}  E_{\langle \ell_1\rangle}(\alpha) \cdot H_{g,w, \ell_2}(\alpha) \cdot g^{M-\ell_2}(\alpha) \cdot \chi \circ g(\alpha)\\
  \nonumber
&=  \left(\sum_{\ell_1 + \ell_2 = \ell}  E_{\langle \ell_1\rangle}(\alpha) \cdot H_{g,w, \ell_2}(\alpha) \cdot g^{M-\ell_2}(\alpha) \right) \cdot \chi \circ g(\alpha)\\  
 &= \chi \circ g(\alpha) \cdot V_\ell(\alpha) , \label{GVg}
\end{align} 
 where we defined:
$$V_\ell(X) = \sum_{\ell_1 + \ell_2 = \ell}  E_{\langle \ell_1 \rangle}(X) \cdot H_{g,w, \ell_2}(X) \cdot g^{M-\ell_2}(X).$$

It is easy to check that:
$$\deg(E \cdot G) \leq d \cdot \left( \frac{q-1}{2} + M \right) + c\cdot q  = D,$$
$$ \deg(V_\ell) \leq h + d \cdot M = u .$$

We now show that for all $\alpha \in \F_q$ and all $\ell$ with $0 \leq \ell < M$:
\begin{align}
 \label{eq:GrV}
 (E \cdot G)^{[\ell]}(\alpha) = r(\alpha) \cdot V_\ell(\alpha).
\end{align}

For $\alpha \not\in S$, we have $r(\alpha) = \chi \circ g (\alpha)$, and
so~\eqref{eq:GrV} follows from Equation~\eqref{GVg}.

For $\alpha \in S$, we have that $V_\ell(\alpha) = 0$ (since $E_{\langle \ell_1 \rangle }(\alpha) = E^{[\ell_1]}(\alpha) = 0$ for all $\ell_1 \leq \ell < M$). Thus $(E \cdot G)^{[\ell]}(\alpha) = 0$ too (by Equation~\eqref{GVg}), and
again we get \eqref{eq:GrV} for this case.

Thus the tuple $(E\cdot G, V_0, \ldots, V_{M-1})$ satisfies the Equations~\eqref{eq:linsyst}, as desired.

\subsubsection*{Step 2: Relating $F$ and $G$}

It turns out that whatever $F(X)$ the algorithm finds must be somewhat related to $G(X)$.

First we quickly note that $F(X)$ must be a nonzero polynomial. We know that
at least of one of $F$ and the $U_\ell$ is nonzero. If some $U_\ell(X)$ is nonzero, then using the facts that (1) $\chi \circ g$ has at most $d$ zeroes, (2) $\Delta(\chi \circ g, r) < \frac{1}{8}q$, (3) $U_\ell(X)$ has at most $h < \frac{1}{4}q$ zeroes, we get $(\chi \circ g(\alpha)) \cdot U_\ell (\alpha)$ must be nonzero for some $\alpha$: and Equation~\eqref{eq:linsyst} tells us that $F(X)$ is nonzero.

Now choose $t,k \in \mathbb N$ as follows:
  $$ t = \frac{3}{8} M.$$
  $$ k = e + 4dM.$$
  
We will show that there exist nonzero polynomials $A(X)$, $B(X)$ of the form:
  $$ A(X) = \sum_{i=0}^{t-1} A_i(X) \Lambda^i(X),$$
  $$ B(X) = \sum_{i=0}^{t + c-1} B_i(X) \Lambda^i(X),$$
  with $\deg(A_i) \leq k$, $\deg(B_i) \leq k + 2e$, such that:
  $$ A(X) \cdot F(X) = B(X) \cdot G(X).$$

We do this by dimension counting again.
The number of coefficients of the  $A_i$ and $B_i$ together is:
\begin{align}
\nonumber
N &= (t) \cdot (k+1) + (t+c) \cdot (k+2e + 1)\\ 
\nonumber
&> \frac{3}{8} M \cdot (e + 4dM) +  \frac{7}{8} M \cdot (3e + 4dM)\\
&= 3 M e + 5dM^2. \label{eq:N}
\end{align}
 and we treat them as unknowns.
By imposing at most $N-1$ homogenous linear constraints on these unknowns, we will
ensure that for each $\ell$ with $0 \leq \ell < M$, and each $\alpha \in \F_q\setminus S$:
\begin{align}
 \label{eq:AFBG}
 (A \cdot F)^{[\ell]}(\alpha) = (B \cdot G)^{[\ell]}(\alpha).
\end{align}

This latter equation implies that the polynomial $(A\cdot F - B \cdot G)(X)$ vanishes with multiplicity $M$ at each $\alpha \in \F_q \setminus S$, and thus has at least
$$M\cdot | \F_q \setminus S| \geq M \cdot (q-e) > \left(\frac{7}{8} + \epsilon\right) Mq$$ roots, counting multiplicity. But
the degree of $(A \cdot F - B \cdot G)(X)$ can be bounded as:
\begin{align*}
\deg( A \cdot F - B \cdot G) &\leq \max( \deg( A \cdot F), \deg( B \cdot G)) \\
&\leq \max( tq + D,  (t+ c) q + ((q-1)/2 + M) \cdot d )\\
&\leq \max( tq  + D,  tq + D)\\
&=  \frac{3}{8}Mq +  d \cdot ( (q-1)/2 + M) + cq\\
&<  \left(\frac{7}{8} + \frac{\epsilon }{2} \right) M q  
\end{align*}

 This implies that $(A \cdot F - B\cdot G)(X) = 0$.

It remains to show how to ensure equations~\eqref{eq:AFBG} using at most $N-1$ linear constraints on the coefficients of the $A_i$ and the $B_i$.

First, observe that
for each $\ell$, there are polynomials $A_{\langle \ell \rangle}(X)$, $B_{\langle \ell \rangle}(X)$, of degrees at most $k, k+ 2e$ respectively (and whose coefficients are homogenous linear combinations of the coefficients of $A,B$), given by:
$$ A_{\langle \ell \rangle}(X) = \sum_{i=0}^{t-1} (-1)^iA_i^{[\ell-i]}(X)$$
$$ B_{\langle \ell \rangle}(X) = \sum_{i=0}^{t+c-1} (-1)^{i}B_i^{[\ell-i]}(X)$$
such that for all $\alpha \in \F_q$:
$$ A^{[\ell]}(\alpha) =  A_{\langle \ell \rangle}(\alpha).$$
$$ B^{[\ell]}(\alpha) = B_{\langle \ell \rangle}(\alpha).$$
(Here we used Lemma~\ref{lem:hasse-xq} to compute $A^{[\ell]}$ and $B^{[\ell]}$ and evaluate them
at $\alpha \in \F_q$).

Next, recall that $F^{[\ell]}(\alpha) = r(\alpha) U_\ell(\alpha)$
for all $\alpha \in \F_q$ and all $\ell$ with $0 \leq \ell < M$.
Then:
\begin{align}
\nonumber
  (A \cdot F)^{[\ell]}(\alpha) &= \sum_{\ell_1 + \ell_2 = \ell} A^{[\ell_1]}(\alpha) F^{[\ell_2]}(\alpha) \\
&= r(\alpha) \cdot \sum_{\ell_1 + \ell_2 = \ell} A_{\langle \ell_1 \rangle}(\alpha) U_{\ell_2}(\alpha).
\label{eq:AF}
\end{align}

Similarly, $G^{[\ell]}(\alpha) = \chi\circ g(\alpha) \cdot H_{g, w, \ell}(\alpha) \cdot g^{M-\ell}(\alpha)$ for all $\alpha \in \F_q$ and all $\ell$ with $0 \leq \ell < M$.
Then:
\begin{align}
\nonumber
  (B \cdot G)^{[\ell]}(\alpha) &= \sum_{\ell_1 + \ell_2 = \ell} B^{[\ell_1]}(\alpha) G^{[\ell_2]}(\alpha) \\
&= \chi \circ g(\alpha) \cdot \sum_{\ell_1 + \ell_2 = \ell} B_{\langle \ell_1\rangle }(\alpha) H_{g,w,\ell_2}(\alpha)\cdot g^{M-\ell_2}(\alpha).
\label{eq:BG}
\end{align}

We can now write down the linear constraints that we impose on the coefficients of the $A_i$ and the $B_i$. 
We ask that for each $\ell$ with $0 \leq \ell < M$, the following equality of polynomials holds: 
\begin{align}
\sum_{\ell_1 + \ell_2 = \ell} A_{\langle \ell_1 \rangle }(X) U_{\ell_2}(X) = \sum_{\ell_1 + \ell_2 = \ell} B_{\langle \ell_1 \rangle }(X) H_{g,w,\ell_2}(X)  g^{M-\ell_2}(X).
\label{eq:babyAFBG}
\end{align}
The coefficients of all these polynomials are homogoeneous linear combinations of the coefficients of the $A_i$ and the $B_i$.

The polynomials on the left hand side of the equality are of degree at most
$$k + u = k + dM + h = 3e + dM.$$
The polynomials on the right hand side of the equality are of degree at most
$$(k + 2e) + dM = 3e + dM.$$
Thus the total number of $\F_q$-linear constraints imposed by these $M$ equalities is at most:
\begin{align*}
 M \cdot (3e + dM+1) <  3Me + 5 d M^2. 
\end{align*}
Combining this with Equation~\eqref{eq:N}, we get that there exist $A(X),B(X)$,
not both $0$, satisfying Equations~\eqref{eq:babyAFBG}.

Finally, from Equations~\eqref{eq:AF},~\eqref{eq:BG} and~\eqref{eq:babyAFBG}, and using the fact that $r(\alpha) = \chi \circ g(\alpha)$ for each $\alpha \in \F_q \setminus S$,
we conclude that Equation~\eqref{eq:AFBG} holds.
This gives us the desired polynomial identity:
$$A(X) \cdot F(X) = B(X) \cdot G(X).$$

\subsubsection*{Step 3: Relating the factor multiplicities of $F$ and $G$}

 Finally, we use the identity and the special form of $A$ and $B$
 to relate the factor multiplicities of $F$ and $G$.
 Lemma~\ref{lem:factormult} below shows that both $A$ and $B$ have the special property that every irreducible factor has multiplicity being approximately a multiple of $q$. On the other hand, the irreducible factors of $g(X)$ appear in $G$ with a multiplicity mod $q$ being approximately $q/2$ (this is the crucial step where we use the squarefreeness of $g(X)$). This lets us identify the nonsquare factors of $g(X)$, and hence $g(X)$. 
 
 \begin{lemma}
 \label{lem:factormult}
 Let $k < q$.
 Let $C(X) \in \F_q[X]$ be a nonzero polynomial of the form:
 $$ \sum_{i=0}^{t-1} C_i(X) \Lambda^i(X),$$
 where $\deg(C_i) \leq k$ for all $i$.
 
 Let $H(X) \in \F_q[X]$ be an irreducible polynomial.
 Let $\mu$ be the highest power of $H(X)$ which divides $C(X)$.
 
 Then $\mu \mod q \in [0, k+t-1]$.
\end{lemma}

\begin{proof}
 Let $\alpha \in \Fqbar$ (the algebraic closure of $\F_q$) be a root of $H(X)$.
 
 Then we have:
 $$\mu = \mult(C, \alpha).$$
 By definition of multiplicity,
 $$C^{[\mu]}(\alpha) \neq 0.$$
 
 We will show that $C^{[\ell]}(X) = 0$ whenever $\ell \mod q \in [k+t, q-1]$.
 This implies that $\mu \mod q \in [0, k+t-1]$.
 
 Suppose $\ell \mod q \in [k+t, q-1]$.
 
 We now compute\footnote{We cannot use Lemma~\ref{lem:hasse-xq} here, since that only talks about derivatives of order $< q$. But it is the same idea.}, using the product rule for Hasse derivatives:
 \begin{align*}
 C^{[\ell]}(X) &= \sum_{i=0}^{t-1} \left( C_i \cdot \Lambda^i \right)^{[\ell]}(X)\\
 &= \sum_{i=0}^{t-1}  \left( \sum_{\substack{\ell_0, \ell_1, \ldots, \ell_i\\\sum \ell_j = \ell}} C_i^{[\ell_0]}(X) \cdot \prod_{j=1}^{i}\Lambda^{[\ell_j]} (X) \right)\\
 \end{align*}

 We now see that all terms in this sum are $0$.
 \begin{itemize}
  \item Every summand with $\ell_0 > k$ is $0$, since each $C_i$ has degree $\leq k$.
 
 \item Every summand where at least one of $\ell_1, \ldots, \ell_j$ does not lie in the set $\{0,1,q\}$ 
 is $0$ (because that derivative of $\Lambda$ is $0$).
 \end{itemize}

 Thus if any summand is nonzero, it must have $\ell_0 \in [0,k]$ and $\ell_1, \ldots, \ell_i \in \{0,1,q\}$.
 
 But if $\ell \mod q \in [k+t, q-1]$,
 then we cannot have $\ell$ written as the sum of 
 one integer in $[0,k]$
 and at most $i \leq t-1$ integers from the set $\{0,1,q\}$.
 
 Thus if $\ell \mod q \in [k+t, q-1]$, $C^{[\ell]}(X) = 0$, as claimed.
\end{proof}

Now let $H(X)$ be any irreducible polynomial in $\F_q[X]$.
Let $\mu_F, \mu_A, \mu_B, \mu_G, \mu_g$ be the highest power of $H(X)$ that divides $F(X), A(X), B(X), G(X), g(X)$ respectively.
By our squarefreeness assumption on $g$, we have that $\mu_g \in \{0, 1\}$.

Then by the identity $AF = BG$,
$$\mu_F = \mu_G + \mu_B - \mu_A.$$
By Lemma~\ref{lem:factormult},
$$ \mu_B \in q\cdot \Z  + [0, 3e + 4dM+ \frac{7}{8}M],$$
$$ \mu_A \in q\cdot \Z + [0, e + 4dM + \frac{3}{8}M].$$
Thus $\mu_F \in \mu_G + q \cdot \Z + [-e -5dM, 3e + 5dM]$.

Note that $\mu_G = ((q-1)/2 + M) \cdot \mu_g$, and thus:
\begin{itemize}
 \item if $\mu_g = 1$, then:
 $$ \mu_G = q/2 +  (M-1)/2,$$
 and thus:
 $$ \mu_F \in q \cdot \Z + q/2 + [-e-5dM, 3e+6dM] \subseteq q \cdot \Z + (3q/8, 7q/8).$$
 \item if $\mu_g = 0$, then:
 $$ \mu_G =0,$$
 and thus:
 $$ \mu_F \in q \cdot \Z + [-e-5dM, 3e+5dM] \subseteq q \cdot \Z + (-q/8, 3q/8).$$
\end{itemize}
Here we used the fact that $e \leq \left( \frac{1}{8} - \epsilon \right)q$ and  $6dM  < 96 \frac{d^2}{\epsilon} < \epsilon q$.

Thus the set $J$ created in Step~\ref{getJstep} ensures that
$\{ H_j \mid j \in J\}$ is exactly the set of all irreducible factors $H(X)$
that divide the squarefree polynomial $g(X)$.

This means that the $f(X)$ found in Step~\ref{getfstep} equals $g(X)$, as desired.
\end{proof}

\subsection{A remark}

An alternate analysis, which is quantitatively slightly weaker and conceptually slightly simpler, is given in Section~\ref{sec:pseudopoly-alternate-A}. It replaces the slightly opaque interpolation argument of Step 2 with another use of a high multiplicity error-locator pseudopolynomial.

\newpage

\section{$\Tr \circ g$}
\label{sec:tr}

Let $q$ be a power of $2$.

Let $g(X) \in \F_q[X]$ be a polynomial of degree at most $d$ with only odd degree monomials:
$$g(X) = a_1 X + a_3 X^3 + \ldots + a_d X^d.$$
In this section, we give an efficient algorithm to recover $g(X)$ given access to a noisy version $r$ of $\Tr \circ g$. 

The main subroutine, Algorithm B, finds the monomial of $g(X)$ of degree $d$. Given that, we can subtract $\Tr$ of that monomial from $r$ and repeat: this recovers the full polynomial $g$.

{\noindent \bf Algorithm B:}\\
{\noindent Parameters: degree $d \leq O(\epsilon \sqrt{q})$ (an odd integer), error-bound $e \leq \left( \frac{1}{8} - \epsilon \right) q$.}\\
{\noindent Input $r: \F_q \to \F_2$ }
\begin{enumerate}
  
 \item Set
 \begin{itemize}
  \item $M = \frac{16}{\epsilon}d$
  \item $c = \frac{1}{2} \cdot M$
  \item $h = 2 \cdot e$
  \item $u = h + dM = 2e + O(\frac{ d^2}{\epsilon} )$.
 \end{itemize}

 \item\label{Blinsyst1} For each $h^* \in [0,h]$
 \begin{itemize}
 \item Try to solve a system of $\F_q$-linear equations to find polynomials $F(X), E(X)$, $U_0(X), U_1(X), \ldots, U_{M-1}(X) \in \F_q[X]$,
 where:
 \begin{itemize}
  \item $F(X)$ has degree exactly $(d/2 + c)q + h^*$.
  \item $E(X)$ has degree exactly $cq + h^*$, and is of the form:
  $$ \sum_{i=0}^{c} E_i(X) \Lambda^{i}(X),$$
  where each $\deg(E_i) \leq h$.
  \item each $U_\ell(X)$ has degree at most $u$.
 
 \item For all $\alpha \in \F_q$, $0 \leq \ell < M$:
 \begin{align}
 \label{eq:Blinsyst1}
  F^{[\ell]}(\alpha)  = r(\alpha) \cdot E^{[\ell]}(\alpha) +  U_\ell(\alpha),
 \end{align}

 \end{itemize}
 \end{itemize}
 
 \item If no such $h^*$ exists, return $0 \cdot X^d$.
 \item Otherwise take one solution $E(X), F(X)$ with degrees $D_E, D_F$.
 
 \item Let $a \cdot X^{D_F}$, $b \cdot X^{D_E}$ be the leading monomials of $F(X), E(X)$ respectively.
 
 Set $a_d = (a/b)^2$.
 
 \item Return $a_d X^d$.
 \end{enumerate}
 \newpage
 
 We now give the main algorithm, which repeatedly invokes the above subroutine.
 
 \noindent{\bf Algorithm C:}\\
 {\noindent Parameters: degree $d \leq O(\epsilon \sqrt{q})$ (an odd integer), error-bound $e \leq \left( \frac{1}{8} - \epsilon \right) q$.}\\
{\noindent Input $r: \F_q \to \F_2$ }

\begin{enumerate}

\item Initialize:
\begin{itemize}
 \item $f(X) = 0$.
 \item $\widetilde{d} = d$.
 \item $\widetilde{r} = r$.
\end{itemize}

\item While $\widetilde{d} > 0$:
\begin{itemize}
 \item Run Algorithm B on $\widetilde{r}$, with parameters $\widetilde{d}, e$.
  \item Let $a_{\widetilde{d}} X^{\widetilde{d}}$ be the monomial it returns.
 \item Update $f(X) = f(X) + a_{\widetilde{d}} X^{\widetilde{d}}$.
 \item For each $\alpha \in \F_q$, update:
 $$\widetilde{r}(\alpha) = \widetilde{r}(\alpha) -\Tr( a_{\widetilde{d}} \alpha^{\widetilde{d}})$$
 \item Update $\widetilde{d} = {\widetilde{d}} -2$.

\end{itemize}
\item Return $f(X)$.
\end{enumerate}

It is clear from the description of the algorithm
that it can be implemented to run in time $\poly(q)$.
In the next subsection show correctness.

\subsection{Correctness of the Algorithm}

Our main theorem about Algorithm $C$ is as follows.

\begin{theorem}
\label{thm:C}
 Let $\epsilon > 0$, and suppose:
 \begin{itemize}
  \item $d \leq \frac{\epsilon}{16} \sqrt{q}$ is an odd integer,
  \item $e = \left( \frac{1}{8} - \epsilon \right) q$.
 \end{itemize}

 Suppose $g(X) \in \F_q[X]$ is a polynomial with degree $\leq d$
 and only odd degree monomials such that $\Delta(\Tr \circ g, r) \leq e$.
 
 Then Algorithm C returns $g(X)$.
\end{theorem}

The key claim is about the behavior of Algorithm B.

\begin{theorem}
\label{thm:B}
 Let $\epsilon > 0$, and suppose:
 \begin{itemize}
  \item $d \leq \frac{\epsilon}{16} \sqrt{q}$ is an odd integer,
  \item $e = \left( \frac{1}{8} - \epsilon \right) q$.
 \end{itemize}
 
 Suppose $g(X) \in \F_q[X]$ is a polynomial with degree $\leq d$
 and only odd degree monomials such that $\Delta(\Tr \circ g, r) \leq e$.

 Then Algorithm B returns $a_d X^d$.
\end{theorem}

Given Theorem~\ref{thm:B}, the correctness of Theorem~\ref{thm:C} is quite easy.

The algorithm maintains the invariant that:
$$\Delta(\widetilde{r}, \Tr \circ (g - f) ) \leq \left(\frac{1}{8} - \epsilon\right) q,$$
and:
$$ \deg(g- f) \leq \widetilde{d}.$$

By the end, we have $f=g$.

We just demonstrate the first step below.
Suppose $g(X)$ has only odd degree monomials and equals
$a_1 X + a_3 X^3 + \ldots + a_d X^d$.
Let $g'(X) = a_1 X + a_3 X^3 + \ldots + a_{d-2} X^{d-2}$.
and $g''(X) = a_d X^d$.
So $g' = g - g''$.

Then Theorem~\ref{thm:B} shows that Algorithm B on input
$r$ returns $a_dX^d$. Next, it updates $\widetilde{r} = r - \Tr \circ g''$. Then we have:
$$ \Delta( \widetilde{r} , \Tr \circ g') = \Delta( r - \Tr \circ g'', \Tr \circ g - \Tr \circ g'') = \Delta( r, \Tr \circ g) \leq \left( \frac{1}{8}-\epsilon\right) q,$$
and 
$\deg(g') = \deg(g - g'') \leq d - 2$,
as desired.

We now prove Theorem~\ref{thm:B} about the correctness of Algorithm B.
\begin{proof}
We begin by handling the (main) case where $a_d \neq 0$.

 The proof has three steps.
 
 Recall that $\TR(X)$ is the Trace polynomial:
 $$ \TR(X) = \sum_{i=0}^{b-1} X^{2^i},$$
 where $q = 2^b$.
 
 Let $G(X) = \TR(g(X))$, which is of degree exactly $d \cdot \frac{q}{2}$
 (because of our assumption $a_d \neq 0$).
 
 First, using the hypothesis that $r$ is close to $\Tr \circ g$, we show that 
 for some $h^*$, the system of linear equations in Step~\ref{Blinsyst1} of Algorithm B
has a nonzero solution.
Similar to the case of the quadratic residue character, we find the polynomials $E_0(X), \ldots, E_c(X)$ so that $E(X)$ is a high multiplicity error-locator pseudopolynomial polynomial, and $V_0(X), \ldots, V_{M-1}(X)$
so that with this choice of the $E_i(X)$,
we have that $F(X) = E(X) \cdot G(X)$ and $U_i(X) = V_i(X)$ satisfy
the linear equations.
Thus the algorithm will find a nonzero solution in that step.

The algorithm may not find the above mentioned ``intended'' solution.
In the second step of the proof, we show that whatever solution  $(F, E, U_0, \ldots, U_{M-1})$ it does find, $F(X)$ and $E(X)$ have some nontrivial relation to $G(X)$.

Finally, we use this nontrivial relation to extract the coefficient of $X^d$ in $g(X)$ from $F$ and $E$, and deduce the result from that.
 
We now proceed with the details.

\subsubsection*{Step 1: The system of linear equations has a nonzero solution.}
We now show that for some $h^* \in [0,h]$, the system of linear equations in Step~\ref{Blinsyst1}
has a nonzero solution.

The key observation is that if the received word $r: \F_q \to \F_2$
had been exaclty $\chi \circ g$, then $F(X) = G(X)$, $E(X) = 1$, $U_\ell(X) = 0$,  would have given a valid solution to the system of equations. Instead, our received word $r$ is merely close to $\Tr \circ g$. So, following the idea of the Berlekamp-Welch decoding algorithm for Reed-Solomon codes, we will zero out $G$ at the error locations to get a valid solution.

Let $S = \{\alpha \in \F_q \mid \Tr \circ g(\alpha) \neq r(\alpha)\}$ be the error set. Let $Z_S(X)$ be the error locator polynomial:
$$ Z_S(X) = \prod_{\alpha \in S} (X-\alpha).$$

Recall that $c = M/2$ and $h = \frac{eM}{c} = 2e$.
We now take $E(X)$ to be a nonzero multiple of $Z_S^M(X)$ of the form:
$$ E(X) = \sum_{i=0}^{c} E_i(X) \Lambda^{i}(X),$$
with $\deg(E_i) \leq h$. 
Such an $E(X)$ exists, since vanishing mod $Z_S^M(X)$ imposes $eM$ constraints on the $(c+1)\cdot(h+1)$-dimensional $\F_q$-linear space $\{(E_0(X), \ldots, E_c(X)) \in \F_q[X]^{c+1} \mid \deg(E_i) \leq h \}$,
and $(c+1) \cdot (h+1) > ch = eM$.

Whatever the degree of $E(X)$ is, it is of the form $iq + j$ for some $i \in [0,c]$ and $j \in [0,h]$. By multiplying this $E(X)$ by $\Lambda^{c-i}(X)$, the degree becomes $cq+j$. We will show that the linear equations have a solution when $h^* = j$.

Exactly as in the case of the quadratic character, this $E(X)$ has two nice properties:
\begin{itemize}
 \item $E^{[\ell]}(\alpha) = 0$ for all $\alpha \in S$ and $\ell$ with $0 \leq \ell < M$.
 \item There is a polynomial $E_{\langle \ell\rangle}(X) \in \F_q[X]$ of degree at most $h = 2e $ such that for all $\alpha \in \F_q$:
 $$ E_{\langle \ell \rangle}(\alpha) = E^{[\ell]}(\alpha).$$

\end{itemize}

We will take our solution $F(X)$ to be:
$$ F(X) = E(X) \cdot G(X),$$ 
(recall that $G(X) = \TR(g(X))$).

Before we compute its derivatives, we need a quick lemma.
\begin{lemma}
Let $g(X) \in \F_q[X]$ be a polynomial 
of degree at most $d$.
Then the $\ell$-th derivative of $G(X) = \TR(g(X))$ is of the form:
 $$(G)^{[\ell]}(X) = \begin{cases}
                      G(X) &  \ell = 0\\
                      H_{g, \ell}(X) & \ell > 0,
                     \end{cases}
                     $$
where $H_{g,\ell}(X) \in \F_q[X]$ is a polynomial of degree at most $d \cdot \ell$.
\end{lemma}
\begin{proof}
The $\ell = 0$ case is immediate.

Now assume $\ell \geq 1$.
 Using the product rule for Hasse derivatives,
 we have:
 \begin{align*}
 (g^{2^i})^{[\ell]}(X) &= \sum_{\ell_1 + \ell_2 = \ell} (g^{2^{i-1}})^{[\ell_1]}(X)\cdot (g^{2^{i-1}})^{[\ell_2]}(X) \\
 &=   \sum_{\ell_1 < \ell /2} 2 \cdot (g^{2^{i-1}})^{[\ell_1]}(X)\cdot (g^{2^{i-1}})^{[\ell - \ell_1]}(X) +
 \begin{cases}
0 & \ell \mbox{ is odd}\\                                                                                                                     \left(\left(g^{2^{i-1}}\right)^{[\ell/2]}(X) \right)^2 & \ell \mbox{ is even}.
\end{cases}\\
&=  \begin{cases}
0 & \ell \mbox{ is odd}\\                                                                                                                     \left(\left(g^{2^{i-1}}\right)^{[\ell/2]}(X) \right)^2 & \ell \mbox{ is even}.
\end{cases}
 \end{align*}
This implies that 
$$ (g^{2^i})^{[\ell]}(X) = 
\begin{cases}
\left(g^{[\ell/2^i]}(X)\right)^{2^i} & 2^i \mbox{ divides } \ell\\
0 & \mbox{ otherwise}.
\end{cases}
$$
 
 If $\ell = 2^s \cdot (2s'+1) \geq 1$, we can then explicitly write:
 $$G^{[\ell]}(X) = H_{g,\ell}(X) = \sum_{i = 0}^{\min(s,b-1)} \left(g^{[\ell/2^i]}(X)\right)^{2^i}.$$

 Inspecting the degrees of the terms in this expression,
 we see that $\deg(H_{g,\ell}) \leq d\cdot 2^s \leq d \cdot \ell $, as claimed.
\end{proof}

Now we compute the derivatives of $E \cdot G$ at $\alpha \in \F_q$:
\begin{align}
\nonumber
(E\cdot G)^{[\ell]}(\alpha)  &= \sum_{\ell_1 + \ell_2 = \ell}  E^{[\ell_1]}(\alpha) \cdot G^{[\ell_2]}(\alpha)\\
\nonumber
  &=  E^{[\ell]}(\alpha) G(\alpha) + \sum_{\substack{\ell_1 + \ell_2 = \ell\\ \ell_1 < \ell}}  E^{[\ell_1]}(\alpha) \cdot G^{[\ell_2]}(\alpha)\\
  \nonumber
  &=  E^{[\ell]}(\alpha) \cdot G(\alpha) + \left(\sum_{\substack{\ell_1 + \ell_2 = \ell\\ \ell_1 < \ell}}  E_{\langle \ell_1\rangle}(\alpha) \cdot H_{g,\ell_2}(\alpha) \right)\\
 &= \Tr \circ g(\alpha) \cdot E^{[\ell]}(\alpha) + V_{\ell}(\alpha), \label{BGVg}
\end{align} 
 where we defined:
$$V_\ell(X) = \sum_{\substack{\ell_1 + \ell_2 = \ell\\ \ell_1 < \ell}}  E_{\langle \ell_1 \rangle}(X) \cdot H_{g, \ell_2}(X) .$$

Then:
$$ \deg(V_\ell) \leq h + d \cdot M = u .$$

Since $\deg(G) = d \cdot \left( \frac{q}{2} \right)$, we get:
$$\deg(F) = \deg(E \cdot G) = (cq + h^*) + d\cdot \frac{q}{2},$$

We now show that for all $\alpha \in \F_q$ and all $\ell$ with $0 \leq \ell < M$:
\begin{align}
 \label{eq:BGrV}
 F^{[\ell]}(\alpha) = (E \cdot G)^{[\ell]}(\alpha) = r(\alpha) \cdot E^{[\ell]}(\alpha) + V_\ell(\alpha).
\end{align}

For $\alpha \not\in S$, we have $r(\alpha) = \Tr \circ g (\alpha)$, and
so~\eqref{eq:BGrV} follows from Equation~\eqref{BGVg}.

For $\alpha \in S$, we have that $E^{[\ell]}(\alpha) = 0$ and $V_\ell(\alpha) = 0$ (since $E_{\langle \ell_1 \rangle }(\alpha) = E^{[\ell_1]}(\alpha) = 0$ for all $\ell_1 \leq \ell < M$). Thus $(E \cdot G)^{[\ell]}(\alpha) = 0$ too (by Equation~\eqref{BGVg}), and
again we get \eqref{eq:BGrV} for this case.

Thus $F = E\cdot G$,  $E$, $U_0 = V_0, \ldots, U_{M-1} = V_{M-1}$ satisfies the Equations~\eqref{eq:Blinsyst1}, as desired.

\subsubsection*{Step 2: Relating $F$, $E$, and $G$}

Let $F(X)$, $E(X)$ be the polynomials that the algorithm found.
By our constraints on the degrees, $\deg(F) - \deg(E) = dq/2$.
It turns out $F(X)$ and $E(X)$ must be somewhat related to $G(X)$.

Concretely, choose $t,k \in \mathbb N$ as follows:
  $$ t = \frac{3}{8} M.$$
  $$ k = e + 4dM.$$
  
We will show that there exist nonzero polynomials $A(X)$, $B(X)$ of the form:
  $$ A(X) = \sum_{i=0}^{t-1} A_i(X) \Lambda^{i}(X),$$
  $$ B(X) = \sum_{i=0}^{t + c + \frac{d}{2}-1} B_i(X) \Lambda^{i}(X),$$
  with $\deg(A_i) \leq k$, $\deg(B_i) \leq k + h + dM$, such that:
  $$ A(X) \cdot \bigg(F(X)  - E(X) \cdot G(X)\bigg) = B(X).$$

The proof is analogous to what we did in the quadratic residue character case; namely, setting up a system of homogenous $\F_q$-linear equations with more unknowns than constraints. 

The total number of unknown coefficients of $A(X), B(X)$ equals:
\begin{align}
\nonumber
 N &= t(k+1) + (t+c+ (d/2)) \cdot (k + h + dM + 1)\\
 \nonumber
 &>  \frac{3}{8} M \cdot \left(e + 4dM \right) + \frac{7}{8} M \cdot \left( 3e + 5dM\right)\\
 &> 3 eM + 5dM^2.
 \label{eq:BN}
\end{align}

By imposing at most $N-1$ homogenous linear constraints on these unknowns, we will
ensure that for each $\ell$ with $0 \leq \ell < M$, and each $\alpha \in \F_q\setminus S$:
\begin{align}
 \label{eq:AFEGB}
 \left(A \cdot (F - E \cdot G) \right)^{[\ell]}(\alpha) = B^{[\ell]}(\alpha).
\end{align}
  
This latter equation implies that the polynomial $A(X)\cdot (F(X) - E(X) \cdot G(X)) - B(X)$ vanishes with multiplicity at least $M$ at each $\alpha \in \F_q \setminus S$, and thus has at least
$$M\cdot | \F_q \setminus S| \geq M \cdot (q-e) > \left(\frac{7}{8} + \epsilon\right) Mq$$ roots, counting multiplicity. But
the degree of $A(X)\cdot (F(X) - E(X) \cdot G(X)) - B(X)$ can be bounded as:
\begin{align*}
\deg( A \cdot (F - EG) - B) &\leq \max( \deg(A) + \deg(F-EG), \deg( B )) \\
&\leq \max( tq +  ((c + (d/2))q + h), (t+c + (d/2))q )\\
&\leq  (t +c+(d/2))q + h \\
&= \frac{7}{8} Mq + (d/2)q + 2e \\
&<  \left(\frac{7}{8} + \frac{\epsilon}{2} \right) Mq.
\end{align*}

 This implies that $A(X) \cdot (F(X) - E(X)\cdot G(X)) = B(X)$.

It remains to show how to ensure Equations~\eqref{eq:AFEGB} using at most $N-1$ linear constraints on the coefficients of the $A_i$ and the $B_i$.

First, observe that
for each $\ell$, there are polynomials $A_{\langle \ell \rangle}(X)$, $B_{\langle \ell \rangle}(X)$, of degrees at most $k, k+h + dM  $ respectively (and whose coefficients are homogenous linear combinations of the coefficients of $A,B$), given by:
$$ A_{\langle \ell \rangle}(X) = \sum_{i=0}^{t-1} (-1)^{i}A_i^{[\ell-i]}(X) $$
$$ B_{\langle \ell \rangle}(X) = \sum_{i=0}^{t+c + (d/2)-1} (-1)^i B_i^{[\ell-i]}(X) $$
such that for all $\alpha \in \F_q$:
\begin{align} A^{[\ell]}(\alpha) = A_{\langle \ell \rangle}(\alpha)\\
 B^{[\ell]}(\alpha) = B_{\langle \ell \rangle}(\alpha). \label{eq:trBpseudo}
 \end{align}
(Again, we used Lemma~\ref{lem:hasse-xq}).

Next, observe that our assumed form on $E(X)$ implies the existence of similar low degree polynomials $E_{\langle \ell \rangle}(X) \in \F_q[X]$ whose evaluations match $E^{[\ell]}$ on $\F_q$.
Namely, (recalling the definition of the $E_i$ from Algorithm B) define:
$$ E_{\langle \ell \rangle}(X) = \sum_{i=0}^{c} (-1)^{i} E^{[\ell-i]}_i(X),$$
and notice that $\deg(E_{\langle \ell \rangle}) \leq h$, and:
$$E^{[\ell]}(\alpha)  = E_{\langle \ell \rangle}(\alpha).$$

Finally, recall that $F^{[\ell]}(\alpha) = r(\alpha) E^{[\ell]}(\alpha) + U_\ell(\alpha)$
for all $\alpha \in \F_q$ and all $\ell$ with $0 \leq \ell < M$.

Thus, for any $\alpha \in \F_q \setminus S$:
\begin{align}
\nonumber
  (F - E\cdot G)^{[\ell]}(\alpha) &= \left( r(\alpha) E^{[\ell]}(\alpha) + U_{\ell}(\alpha) \right) - \sum_{\ell_1 + \ell_2 = \ell} E^{[\ell_1]}(\alpha) G^{[\ell_2]}(\alpha) \\
  \nonumber
&= \left( r(\alpha) \cdot E_{\langle \ell \rangle}(\alpha) + U_{\ell}(\alpha) \right) -
\left( E_{\langle \ell \rangle}(\alpha) G(\alpha) + \sum_{\substack{\ell_1 + \ell_2 = \ell\\ \ell_1 < \ell}} E_{\langle \ell_1 \rangle}(\alpha) H_{g,\ell_2}(\alpha)\right)\\
\nonumber
&= \left( r(\alpha) - G(\alpha) \right) \cdot E_{\langle \ell \rangle}(\alpha) + \left( U_\ell(\alpha) - \sum_{\substack{\ell_1 + \ell_2 = \ell\\ \ell_1 < \ell}} E_{\langle \ell_1 \rangle}(\alpha) H_{g,\ell_2}(\alpha)\right)\\
\nonumber
&= 0 + W_\ell(\alpha)\\
&= W_\ell(\alpha),
\end{align}
 where $W_\ell(X)$ is the polynomial of degree at most $h+ dM$ given by:
 $$W_\ell(X) = U_\ell(X) - \sum_{\substack{\ell_1 + \ell_2 = \ell\\ \ell_1 < \ell}} E_{\langle \ell_1 \rangle}(X) H_{g,\ell_2}(X).$$
 and we used the fact that $r(\alpha) = G(\alpha)$ for all $\alpha \in \F_q \setminus S$.

 So for any $\alpha \in \F_q \setminus S$, we have:
 \begin{align}
\nonumber
  (A \cdot (F-E\cdot G) )^{[\ell]}(\alpha) &= \sum_{\ell_1 + \ell_2 = \ell} A^{[\ell_1]}(\alpha) \left( F - E \cdot G  \right)^{[\ell_2]}(\alpha)\\
  &=\sum_{\ell_1 + \ell_2 = \ell} A_{\langle \ell_1 \rangle}(\alpha) \cdot W_{\ell_2}(\alpha).\label{eq:BAF}
 \end{align}

We can now write down the linear constraints that we impose on the coefficients of the $A_i$ and the $B_i$. 
We ask that for each $\ell$ with $0 \leq \ell < M$, the following equality of polynomials holds: 
\begin{align}
\sum_{\ell_1 + \ell_2 = \ell} A_{\langle \ell_1 \rangle }(X) W_{\ell_2}(X) = B_{\langle \ell \rangle}(X).
\label{eq:babyAFEGB}
\end{align}
The coefficients of all these polynomials are homogenous linear combinations of the coefficients of the $A_i$ and the $B_i$.

The polynomials on the left hand side of the equality are of degree at most
$$k + (h+ dM) = 3e + dM.$$
The polynomials on the right hand side of the equality are of degree at most
$$k+ h + dM = 3 e + dM.$$
Thus the total number of $\F_q$-linear constraints imposed by these $M$ equalities is at most:
\begin{align*}
 M \cdot ( 3e + dM + 1) = 3 eM + dM^2 + M <  3eM + 5 d M^2. 
\end{align*}
Combining this with Equation~\eqref{eq:BN}, we get that there exist $A(X),B(X)$,
not both $0$, satisfying Equations~\eqref{eq:babyAFEGB}.

Finally, from Equations~\eqref{eq:trBpseudo},~\eqref{eq:BAF} and~\eqref{eq:babyAFEGB},
we conclude that Equation~\eqref{eq:AFEGB} holds.
This gives us the desired polynomial identity:
$$A(X) \cdot (F(X) - E(X) G(X) ) = B(X).$$

\subsubsection*{Step 3: Relating the leading coefficients of $F$, $E$ and $G$}

From the above identity, we have:
$$ \deg(A) + \deg( F - E \cdot G) = \deg (B).$$

By choice of $A$ and $B$, we have:
$$\deg(A) \in q \cdot \Z + [0,k] = q \cdot \Z + \left[0, e + 4dM\right] \subseteq q \cdot \Z + \left[ 0, \frac{1}{8}q \right),$$
$$\deg(B) \in q \cdot \Z + \left[0,k+ h + dM \right] \subseteq q \cdot \Z + \left[0, 3e + 5dM\right) \subseteq q \cdot \Z + \left[0, \frac{3}{8}q \right).$$

Thus
$$\deg( F - E \cdot G) = \deg(B) - \deg(A)  \in q \cdot \Z + \left( -\frac{1}{8}q, \frac{3}{8} q \right).$$

Now we study $\deg(F)$ and $\deg(E \cdot G)$.

By design, and using the fact that $d$ is odd:
$$\deg(F) = \deg(E \cdot G) = (c+d/2)q + h^* \in q \cdot \Z + \left[\frac{q}{2}, \frac{q}{2} + 2e \right] \subseteq q \cdot \Z + \left[ \frac{1}{2}q, \frac{3}{4} q \right).$$

Now for the crucial point. 
Since the intervals $\left( -\frac{1}{8}q, \frac{3}{8}q \right)$ and 
$\left(\frac{1}{2} q, \frac{3}{4} q\right)$ are (comfortably\footnote{There is some slack here. The decoding radius can be improved a little bit (but not much) by tuning the choices of $c,h$ in the algorithm and $k,t$ in the analysis, to $ \frac{1}{2} \left(1 - \frac{1}{\sqrt{2}} - \epsilon\right) \cdot q \approx (0.145 - \epsilon )q $, with slightly more complicated calculations and no new conceptual ideas. This improved decoding radius equals $\frac{1}{2}(J(1/2) - \epsilon)\cdot q$, where $J(\delta)$ is the alphabet-free Johnson radius for codes of relative distance $\delta$.}) disjoint mod $q$,
we have that $\deg(F - E \cdot G) \neq \max(\deg(F), \deg(E\cdot G))$.

This means that $F$ and $E\cdot G$ must have the same leading monomial!

We have $F$ and $E$ in our hands with leading coefficients $a, b \neq 0$, and so the leading coefficient of $G$ (which we know equals $a_d^{q/2}$) must equal $\frac{a}{b}$.

Thus:
$$ a_d = a_d^q = ( a_d^{q/2} ) ^2 = \left( \frac{a}{b}   \right)^2,$$
which is as computed by the algorithm.

This shows that Algorithm B correctly returns $a_d X^d$ whenever $a_d \neq 0$.

Finally we have to show what happens if $a_d = 0$.
We show that there is no $h^*$ for which the system of linear equations is solvable (and thus it will return $0 \cdot X^d$, as desired).

If there was a solution $F, E$ found, then $\deg(F) - \deg(E) = dq/2$, and
by our Step 2 and Step 3 of our proof above, $F(X)$ and $E(X) G(X)$ have the same leading monomial. However $a_d = 0$ implies that $\deg(G) < dq/2 = \deg(F) - \deg(E)$, a contradiction.

This completes the proof of correctness of Algorithm B.
\end{proof}

\subsection{A remark}

 Recall Step 2 of the proof of correctness of Algorithm B, where we deduced that $G$ was related to $F,E$ via a relation of the form:
 $$ A(X) \cdot (F(X) - E(X) \cdot G(X)) = B(X),$$
 with $A(X),B(X)$ of a special form.
 
 It turns out that this information already implies that $g$ is uniquely deterimined by $F, E$.
  However, we do not know how to {\em efficiently} find $g$ from just $F$ and $E$. This is why we had to take the longer route of just getting the leading coefficient of $g$, and then repeating (which involves solving for a new $F, E$).

 Why do $F$ and $E$ uniquely determine $g$? Suppose not,
 and that $g_1$, $g_2$ are distinct polynomials with only odd degree monomials such that $G_1 = \TR(g_1(X)), G_2 = \TR(g_2(X))$ both satisfy equations as above:
  $$ A_1(X) ( F(X) - E(X) G_1(X)) = B_1(X),$$
 $$ A_2(X) ( F(X) - E(X) G_2(X)) = B_2(X),$$
 Eliminating $F(X)$, we get:
 $$ A_1(X) A_2(X) E(X) \bigg(G_1(X) - G_2(X) \bigg) + (B_2(X) A_1(X) - B_1(X) A_2(X)) = 0.$$
 
 This kind of polynomial was explicitly considered by Stepanov, and he showed that it cannot equal $0$.
 
 The key point is that $G_1(X) - G_2(X)$ is a nonzero polynomial with degree in $q \cdot \Z + q/2$.
 
 This means that the two polynomials
  $A_1(X) A_2(X) E(X) \cdot \bigg(G_1(X) - G_2(X) \bigg)$ and $B_1(X) A_2(X) - B_2(X) A_1(X)$ cannot have equal degrees: the former has degree in the interval $[q/2, q/2 + 2k + h] \mod q$, while the latter has degree in the interval $[0, 2k+h] \mod q$, and these are disjoint intervals mod $q$.
  
  Therefore their sum cannot be zero, and we get a contradiction to the assumption that $g$ was not uniquely determined.

\newpage

\section{Theory of Pseudopolynomials}
\label{sec:pseudopoly}

The results of the previous sections made extensive use of high degree polynomials, all of whose derivatives behaved like low degree polynomials on $\F_q$.

We call these objects pseudopolynomials. In this section, we formally define them and build up some basic theory. We hope this will be useful for other algebraic and algorithmic problems involving polynomials over finite fields.

\begin{definition}[Pseudoderivative]
Let $A(X) \in \F_q[X]$. 

Define the {\em $\ell$-th pseudoderivative} of $A$, denoted $A_{\langle \ell \rangle}(X)$,  to be 
the unique polynomial of degree at most $q-1$ whose evaluations on $\F_q$ agree with the evaluations of $A^{[\ell]}(X)$. 
Explicitly:
$$ A_{\langle \ell \rangle}(X) = \left( A^{[\ell]}(X) \mod \Lambda(X) \right),$$
where $\Lambda(X) = X^q - X$.

Abusing notation, we will sometimes use $A_{\langle \ell \rangle}$
to denote the function
$$A_{\langle \ell \rangle} :  \F_q \to \F_q,$$
defined by $A_{\langle \ell \rangle}(\alpha) = A^{[\ell]}(\alpha)$.
\end{definition}

\begin{definition}[Pseudodegree]
Let $A(X) \in \F_q[X]$.

We define the {\em pseudodegree} of $A$,
denoted $\pdeg(A)$ to be:
$$ \max_{\ell \geq 0} \deg(A_{\langle \ell \rangle}).$$

\end{definition}

We say $A(X)$ is a {\em $k$-pseudopolynomial} if
$$ \pdeg(A) \leq k.$$ 

With this notation in place, we now list an assortment of useful facts
about pseudopolynomials.

\subsection{Basic Properties}
\label{sec:pseudopoly-basic}

With this notation in place, we now list an assortment of basic properties of pseudopolynomials. 

\begin{itemize}
\item $\pdeg(A + B) \leq \max (\pdeg(A), \pdeg(B))$.

\item $\pdeg(A \cdot B) \leq \pdeg(A) + \pdeg(B)$.

\item We have the following algebraic characterization of pseudodegree.

\begin{restatable}{lemma}{algchar}
\label{lem:pdeg-baselambda}
  Let $A(X) \in \F_q[X]$, and
let 
$$ A(X) = \sum_{i \geq 0} A_i(X) \Lambda^i(X),$$
with each $\deg(A_i) < q$,
be the base-$\Lambda(X)$ expansion of $A(X)$.

Then
$$\pdeg(A) = \max_{i\geq 0} \deg(A_i).$$
\end{restatable}

The proof appears in Appendix~\ref{sec:alg-char}.

\item For polynomials $A(X)$ of moderate degree, the pseudodegree can be approximately expressed in terms of the standard monomial expansion of $A(X)$.

If $A$ has pseudodegree $\leq k$ and degree $D$, then $A(X)$ can be written in the form $$\sum_{i = 0}^{\lfloor D/q\rfloor} \widehat{A}_i(X) X^{iq},$$
 where $\deg(\widehat{A}_i(X)) \leq k + \lfloor D/q \rfloor$.
 In the reverse direction, if $A(X)$ can be written in the form
 $$\sum_{i=0}^t \widehat{A}_i(X) X^{iq},$$
 with $\deg(\widehat{A}_i) \leq k$,
 then $A(X)$ has pseudodegree at most $k + t$ and degree at most $tq +k$.
 
 For all the applications in this paper, we could have worked with the standard monomial basis with negligible losses.
 
 \item Suppose $A(X) \in \F_q[X]$  with $\deg(A) < Mq$.
 Suppose $\deg(A_{\langle \ell \rangle}) \leq k$ for all $\ell$ with $0\leq \ell < M$.
 Then $A$ has pseudodegree $\leq k$.

 That is, to check that $A$ has small pseudodegree, it suffices to check low degreeness of $A_{\langle \ell \rangle}$ for $\ell$ up to $\lceil\deg(A)/q\rceil$. 
 
 This follows from the proof of Lemma~\ref{lem:pdeg-baselambda}.
 
 \end{itemize}
 
 \subsection{Multiplicities}
 
 \begin{itemize}
 \item Suppose $A(X) \in \F_q[X]$ has pseudodegree $\leq k$ and degree at most $D$.
 Let $H(X) \in \F_q[X]$ be irreducible, and let $\mu$ be the highest power of $H(X)$ that divides $A(X)$. Then:
 $$\mu \mod q \in \left[0, k + \frac{D}{q}\right].$$
 
 This is Lemma~\ref{lem:factormult}.

 \item Suppose $A(X) \in \F_q[X]$ has pseudodegree $\leq k$.
 Then:
 $$\deg(A) \mod q \in [0,k].$$
 
 This is immediate from Lemma~\ref{lem:pdeg-baselambda}.
 
 This is the analogue of the previous fact for the place at infinity.
 
\end{itemize}

\subsection{The number of high multiplicity zeroes}
 
\begin{itemize} 
 \item For any $S \subseteq \F_q$, and any $c,k, M$ satifying:
 $$ |S| < \frac{c}{M} \cdot k,$$
 there exists a nonzero $A(X) \in \F_q[X]$ with $\deg(A) < cq$ and $\pdeg(A) < k$ such that for all $\alpha \in S$,
 $$\mult(A, \alpha) \geq M.$$
 
 This follows from dimension and constraint counting.
 We saw this in the construction of the error-locating pseudopolynomials in Step 1 of the proofs of both Theorem~\ref{thm:A} and Theorem~\ref{thm:B}.
 
 \item In the reverse direction, the next theorem shows that when $q$ is prime, every
 nonzero $A(X)$ with $\deg(A) < cq$ and $\pdeg(A) < k$ has at most $2 \frac{c}{M} k$ roots of multiplicity at least $M$. The bound improves to almost $(1+\epsilon)\frac{c}{M} k$ when $c < O(\epsilon) M$ or $c > (1-O(\epsilon)) M$.

 \begin{restatable}{theorem}{pseudoSZ}
 \label{thm:pseudoSZ}
 Let $q$ be prime.
 
 Let $A(X) \in \F_q[X]$ be a nonzero polynomial with
 $\deg(A) < cq$ and $\pdeg(A) < k$.
 Let $M$ be such that $c < M < q$.
 
 Then:
 $$ | \{ \alpha \in \F_q \mid \mult(A, \alpha) \geq M \}| \leq \min\left(  \frac{c}{M-c+1} \cdot k + c,\quad k \right).$$
\end{restatable}

 The proof uses an argument of Guruswami and Kopparty~\cite{GK-subspace-design} on ranks of Wronskians, originally discovered in the context of subspace designs and multiplicity codes.
 
 It seems plausible that there is an upper bound much closer to $\frac{c}{M} \cdot k$. We do not know whether the primality of $q$ is needed for such a statement.

 \item There is a natural error-correcting code here, closely related to multiplicity codes. There are 4 governing parameters: $q, M, c, k$. 
 
 Let $\Sigma = \F_q^M$. The codewords of this code will lie in $\Sigma^{\F_q}$. For each $A(X) \in \F_q[X]$ with $\deg(A) < cq$ and $\pdeg(A) < k$, we define the codeword $y_A : \F_q \to \Sigma$ by:
 $$ y_A(\alpha) = ( A^{[0]}(\alpha), A^{[1]}(\alpha), \ldots, A^{[M-1]}(\alpha)).$$
 
 This code has cardinality $|\Sigma|^{ck/M}$, block length $q$, rate $R = \frac{c}{M} \cdot \frac{k}{q}$, and minimum distance at least
 $$\left( 1 - \min\left( \frac{k}{q},\quad  \frac{c}{M-c+1}\cdot\frac{k}{q} + \frac{c}{q}, \quad \frac{c}{M}   \right) \right) \cdot q,$$
 which is always at least $(1 - 2R) \cdot q$.
 
 When $k = q$, this is a multiplicity code, and the distance of the code is $(1-R) \cdot q$. When $c = M$ this is an interleaved Reed-Solomon code, and the distance of the code is $(1-R) \cdot q$. Maybe the distance of this code is always $(1-R)q$?

 \end{itemize}

 \subsection{Twisted pseudopolynomials}
 
 Towards presenting (a mild varation of) Stepanov's proof of the Weil bound for multiplicative character sums, we now define twisted pseudopolynomials and demonstrate a clean version of the interpolation argument that we saw in Section~\ref{sec:chi}.
 
 \begin{itemize}

 \item Suppose $r: \F_q \to \F_q$ is a function. Let $k < q$.
 We say $F(X) \in \F_q[X]$ is an $r$-twisted $(h,M)$-pseudopolynomial, if for all $\ell < M$, there is some $U_\ell(X) \in \F_q[X]$ of degree at most $h$ such that:
 $$ F_{\langle \ell \rangle} = r \cdot U_\ell.$$

 \item The following lemma shows that any two $r$-twisted $(h,M)$-pseudopolynomials 
 with degree at most $cq$ are very closely related, provided $h$ and $c$ are small enough.
 \begin{lemma}
 \label{lem:pure2}
  Let $c,h, M$ be parameters, with $c < M/2$.
  Let $r : \F_q \to \F_q$.
  
  Suppose $F(X), G(X) \in \F_q[X]$ with $\deg(F), \deg(G) < cq$
  are both $r$-twisted $(h,M)$-pseudopolynomials.

  Suppose $$ k > \frac{M}{M-2c} \cdot (h+1).$$
  
  Then there exist nonzero $k$-pseudopolynomials $A(X), B(X)$,
  with $\deg(A), \deg(B) < Mq$, such that:
  $$ A(X) \cdot F(X) = B(X) \cdot G(X).$$
 \end{lemma}
 \begin{proof}
 If $k \geq q$ then the result is trivial, so we may assume that $k < q$.
 
 Suppose for each $\ell < M$, we have:
 $$F_{\langle \ell \rangle} = r \cdot U_\ell,$$
 $$G_{\langle \ell \rangle} = r \cdot V_\ell,$$
 where $U_\ell(X), V_\ell(X) \in \F_q[X]$ have
 degrees at most $h$.
 
  We search for the polynomials $A, B$ of the form:
  $$ A(X) = \sum_{i=0}^{M-c-1} A_i(X) \Lambda^i(X),$$
  $$ B(X) = \sum_{i=0}^{M-c-1} B_i(X) \Lambda^i(X),$$
  where $\deg(A_i), \deg(B_i) \leq k$.
  
  The number of unknowns is strictly greater than $2 \cdot (M-c) \cdot k$.
  
  We now apply homogeneous $\F_q$-linear constraints to these unknowns. These constraints are meant to express the equality:
  $$ (A \cdot F)_{\langle \ell \rangle} = (B \cdot G)_{\langle \ell \rangle},$$
  for each $\ell < M$.
  Both sides of this can be expanded in terms of the pseudoderivatives of $A,B,F,G$, the latter two of which can be written in terms of $r \cdot U_\ell$ and $r \cdot V_\ell$.
  
  All these constraints can be written compactly in the ring $\F_q[X,T]$ as:
  $$ \left( \sum_{\ell_1 < M} A_{\langle \ell_1 \rangle}(X) \cdot T^{\ell_1} \right) \cdot 
  \left( \sum_{\ell_2 < M} U_{\ell_2}(X) \cdot T^{\ell_2} \right) =
  \left( \sum_{\ell_3 < M} B_{\langle \ell_3 \rangle}(X) \cdot T^{\ell_3} \right) \cdot 
  \left( \sum_{\ell_4 < M} V_{\ell_4}(X) \cdot T^{\ell_4} \right) \mod T^M
  $$
  
  There are $M$ equalities of polynomials of degree at most $h + k$ in $\F_q[X]$,
  each having coefficients being homogenous linear combinations of the coefficients of the $A_i$ and $B_i$.
  Thus there are a total of $(h+k+1) \cdot M$ $\F_q$-linear constraints.
  
  By choice of $k$, we have:
  $$ (h + k+1) \cdot M  <  2 k \cdot (M-c),$$
  and so this system of equations has a nonzero solution.
  
  Finally, we show that $A\cdot F = B \cdot G$.
  By design of $A$ and $B$, for each $\ell < M$ we have:
  $$ \left(A \cdot F - B \cdot G\right)_{\langle \ell \rangle} = 0.$$
  Thus $\mult(A\cdot F - B \cdot G, \alpha) \geq M$ for each $\alpha \in \F_q$.
  On the other hand,
  $$\deg(A \cdot F - B \cdot G) < c q + (M-c)q  < Mq.$$
  This means that $A \cdot F = B \cdot G$, as desired.

 \end{proof}
 
\item  The next lemma shows that if we have many $r$-twisted $(h,M)$-pseudopolynomials 
 with degree at most $cq$ with $c,h$ small, then they are nontrivially related.
 The smallness requirement on $c$ is weaker, but the deduced relation is also weaker.
 
 \begin{lemma}
 \label{lem:pureb}
  Let $m,c,h, M$ be parameters, with $c < \frac{m-1}{m}M$.
  Let $r : \F_q \to \F_q$.
  
  Suppose $F_1(X) ,\ldots, F_m(X) \in \F_q[X]$ with $\deg(F_i) < cq$
  are all $r$-twisted $(h,M)$-pseudopolynomials.

  Suppose $$ k > \frac{M}{(m-1)M-mc} \cdot (h+1).$$
  
  Then there exist $k$-pseudopolynomials $A_1(X), \ldots, A_m(X)$
  with $\deg(A_i) < Mq$ and not all zero, such that:
  $$ \sum_i A_i(X) F_i(X) = 0.$$
 \end{lemma}

 The proof is exactly like that of Lemma~\ref{lem:pure2}, which is the $m = 2$ case.

\item The exact same argument also gives us robust versions.

\begin{lemma}
\label{lem:impure}
  Let $c,h,\gamma, M$ be parameters, with $c < (1- 2\gamma)M/2$.
  Let $r_1, r_2 : \F_q \to \F_q$,
  with:
  $$ \Delta(r_1, r_2) \leq \gamma q.$$
  
  Suppose $F(X), G(X) \in \F_q[X]$ with $\deg(F), \deg(G) < cq$
  which are $r_1$-twisted and $r_2$-twisted $(h,M)$-pseudopolynomials respectively.

  Suppose $$ k > \frac{M}{(1-2\gamma)M-2c} \cdot (h+1).$$
  
  Then there exist nonzero $k$-pseudopolynomials $A(X), B(X)$,
  with $\deg(A), \deg(B) < Mq$, such that:
  $$ A(X) \cdot F(X) = B(X) \cdot G(X).$$
 \end{lemma}
The only change to the proof is that now we search for $A(X)$ and $B(X)$ of degree less than $((1-\gamma)M - c)q$, so that the fewer agreements we are given still translate into a polynomial identity.

 \begin{lemma}
 \label{lem:impureb}
  Let $m,c,h,\gamma, M$ be parameters, with $c < \frac{(m-1) - m\gamma}{m}M$.
  Let $r_1, \ldots, r_m : \F_q \to \F_q$,
  with:
  $$ |\{\alpha \in \F_q \mid   r_i(\alpha) \neq r_j(\alpha) \mbox{ for some $i,j$}\}| \leq \gamma q .$$
  
  Suppose $F_1(X) ,\ldots, F_m(X) \in \F_q[X]$ with $\deg(F_i) < cq$
  are $r_i$-twisted $(h,M)$-pseudopolynomials.

  Suppose $$ k > \frac{M}{(m-1 -m \gamma)M-mc} \cdot (h+1).$$
  
  Then there exist $k$-pseudopolynomials $A_1(X), \ldots, A_m(X)$,
  with $\deg(A_i) < Mq$, not all zero, such that:
  $$ \sum_i A_i(X) F_i(X) = 0.$$
 \end{lemma}

\end{itemize}

\subsection{The Weil bounds}

\begin{itemize}
 
\item
We can now prove the Weil bound for the quadratic residue character.

Let $f(X), g(X) \in \F_q[X]$ be distinct monic squarefree polynomials of degree at most $d \leq O(\epsilon \sqrt{q})$. We will show that
$$\Delta(\chi \circ f, \chi \circ g) \geq \left(\frac{1}{2} - \epsilon\right) q.$$

Set $M = \frac{2}{\epsilon} d$, $c = d$.
Let $F(X)= f(X)^{(q-1)/2 + M}$, 
$G(X) = g(X)^{(q-1)/2 + M}$.

Then $F$ and $G$ are $(\chi \circ f)$-twisted and $(\chi \circ g)$-twisted $(dM, M)$-pseudopolynomials respectively.

Now set $h = dM$ and $\gamma = \frac{1}{2} -\epsilon$.
If $\Delta(\chi \circ f, \chi \circ g) \leq \gamma q$,
then Lemma~\ref{lem:impure} tells us that for:
$$k = \frac{1}{\epsilon} \cdot (h+1) \leq O(\frac{d^2}{\epsilon^2}) < \frac{q}{4},$$
there exist nonzero $k$-pseudopolynomials $A(X), B(X)$, with $\deg(A), \deg(B) < Mq$, such that:
$$ A(X) \cdot F(X) = B(X) \cdot G(X),$$

Let $H(X)$ be a polynomial that divides one of $f(X), g(X)$ but not the other.

Lemma~\ref{lem:factormult} applied to $H(X)$ and the above identity then gives us a contradiction.

\item Next we prove the Weil bound for the $m$-th power residue  character $\chi_m$, for $m$ prime.

Suppose $f,g \in \F_q[X]$ are of degree at  most $d \leq O_m(\epsilon \sqrt{q})$, and are distinct, monic, with each irreducible factor appearing with multiplicity $\in \{1,2, \ldots, m-1\}$. This last condition generalizes squarefreeness in the case $m = 2$. We will show that
$$ \Delta( \chi_m\circ f, \chi_m \circ g) \geq \left( 1- \frac{1}{m} - \epsilon \right) q.$$

Let $M = O_m(\frac{d}{\epsilon})$, $c = O_m(d)$.
Let $F(X) = f^{(q-1)/m+M}(X)$, $G(X) = g^{(q-1)/m+M}(X)$.

Let $r_1, \ldots, r_m: \F_q \to \F_q$ be given by:
$$ r_i(\alpha) = (\chi_m^{i-1} \circ f(\alpha)) \cdot (\chi_m^{m-i} \circ g(\alpha)).$$

Let $F_1(X), \ldots, F_m(X)$ be given by:
$$ F_i(X) = F^{i-1}(X) \cdot G^{m-i}(X).$$

Then $F_i$ is a $r_i$-twisted $(O_m(dM), M)$-pseudopolynomial.

Now set $h = O_m(dM)$ and $\gamma = 1 - \frac{1}{m}  - \epsilon$.
If $\Delta(\chi_m \circ f, \chi_m \circ g) \leq \gamma q$, then
$$ |\{\alpha \in \F_q \mid   r_i(\alpha) \neq r_j(\alpha) \mbox{ for some $i,j$}\}| \leq \gamma q ,$$
Then the Lemma~\ref{lem:impureb} tells us that for:
$$k = \frac{1}{\epsilon} \cdot (h+1) \leq O(\frac{d^2}{\epsilon^2}) < \frac{q}{2m},$$
there are $k$-pseudopolynomials $A_i(X)$, with $\deg(A_i) < Mq$, not all zero, such that:
$$ \sum_{i=1}^m A_i(X) \cdot F^{i-1}(X) \cdot G^{m-i}(X) = 0.$$

Let $H(X)$ be an irreducible polynomial that appears with distinct factor multiplicity in $f(X)$ and $g(X)$. Then Lemma~\ref{lem:factormult} tells us that 
all the nonzero terms in the above sum have distinct factor multiplicity of $H(X)$ (because they are distinct mod $q$), and thus the sum cannot be $0$, a contradiction.

\item A similar argument can be given for the Weil bounds for additive character sums over fields of small characteristic. We just give a sketch.

Let $q$ have characteristic $p$. Let $\Tr : \F_q \to \F_q$ and $\TR(X) \in \F_q[X]$ be the field trace map and trace polynomial respectively.

Let $r: \F_q \to \F_q$ be a function.
We call a polynomial $F(X)$ an $r$-guided $(h,M)$-pseudopolynomial if:
for all $\ell < M$, there exists $U_\ell(X) \in \F_q[X]$ with $\deg(U_\ell) \leq h$, such that whenever $\alpha \in \F_q$ satisfies $r(\alpha) = 0$, we have:
$$ F_{\langle \ell \rangle}(\alpha) = U_\ell(\alpha).$$

Let $f(X) \in \F_q[X]$ have degree $d \leq O_p(\epsilon \sqrt{q})$ and all monomials having degree relatively prime to $p$. We will show there are at most $\left(\frac{1}{p} + \epsilon\right) q$ values of $\alpha \in \F_q$ such that
$\Tr\circ f(\alpha) = 0$. (By linearity, this gives us a lower bound of $\left( 1 - \frac{1}{p} - \epsilon \right) q$ on the distance between any two distinct $\Tr \circ f$ and $\Tr \circ g$ of this form).

Let $F(X) = \TR \circ f(X)$.

Then we have that for all $i \in [1, p-1]$, $F^i(X)$ is a $(\Tr \circ f)$-guided $(dM,M)$-pseudopolynomial. Further, $\deg(F^i) = \frac{i \cdot \deg(f)}{p} \cdot q < dq$, and thus the $F^i$ for $i \in [1,p-1]$ have degrees mod $q$ that differ pairwise by at least $\frac{q}{p}$.

Suppose $\Tr \circ f(\alpha) = 0$  for more than $(\frac{1}p + \epsilon)q$ values of $\alpha \in \F_q$.

Now let $A_0(X), \ldots, A_{p-1}(X)$ be unknown polynomials with $\deg(A_i) < \left(\left(\frac{1}{p} + \frac{\epsilon}{2}\right)M - d \right) q$ and $\pdeg(A_i) < O(\frac{dM}{\epsilon}) <  \frac{q}{4p}$.
By counting constraints, we can ensure that the polynomial:
$$ B(X) = \sum_{i=0}^{p-1} A_i(X) \cdot F^i(X),$$
which has degree at most $\left(\frac{1}{p} + \frac{\epsilon}{2}\right)Mq$, vanishes with multiplicity $M$ at all $\alpha \in \F_q$ where $(\Tr \circ f)(\alpha) = 0$. (We use the $(\Tr \circ f)$-guidedness of $F^i$ here.)
This means that the polynomial $B(X)$ equals $0$.

The final contradiction come from the fact that the degrees of the different terms in $B(X)$ are distinct mod $q$, and hence the sum cannot be $0$.
\end{itemize}

\subsection{Alternate analysis of Algorithm A}
\label{sec:pseudopoly-alternate-A}

We now present a slightly different (and possibly more conceptual) analysis of Algorithm A, using some of the language that we just developed. For the case of zero errors this is perhaps the simplest proof of correctness.

This analysis can handle a smaller  (but still $\Omega(q)$) number of errors.
For simplicity of exposition, let we assume that $g(X)$ does not have any degree $1$ irreducible factors -- and thus $\chi \circ g$ is only $\{\pm 1\}$-valued.
We will also assume that $r$ is only $\{ \pm 1 \}$-valued.

We keep notation from the proof of Theorem~\ref{thm:A}, and pick up at the beginning of Step 2 of the proof.
The $F(X)$ found by the algorithm is an $r$-twisted $(h,M)$-pseudopolynomial.
We know that $G(X)$ is an $(\chi \circ g)$-twisted $(O(dM), M)$-pseudopolynomial.

Then $(F \cdot G)(X)$ is an $(r \cdot (\chi \circ g) )$-twisted $(h + O(dM), M)$-pseudopolynomial.

Now observe that the function $r \cdot (\chi \circ g) : \F_q \to \F_q$ is close to the constant $1$ function. It evaluates to $1$ on at least $q - e$ values of $\alpha \in \F_q$. Had it been $1$ valued everywhere, $(F\cdot G)$ would have been a genuine pseudopolynomial.

We can fix it using a high multiplicity error-locator pseudopolynomial, as in Step 1 of the proof. 
Let $c', h'$, be given by $c' = M/4$, $h' = 4e$. 
Then, since $c' \cdot (h'+1) > eM$, we can get a polynomial $E'(X)$ with $\deg(E') < c' q$ and $\pdeg(E') \leq h'$ that vanishes at all points of $S$ with multiplicity $M$. 

Then $(E' \cdot F \cdot G)(X)$ has the property that for all $\ell < M$:
$$\deg((E' \cdot F \cdot G)_{\langle \ell \rangle}(X)) \leq h' + h + O(dM).$$
Furthermore, $\deg(E' \cdot F \cdot G) < c' q + cq + 2d \cdot ((q-1)/2 + M) < Mq$.

By the facts in Section~\ref{sec:pseudopoly-basic}, this means that $(E \cdot F \cdot G)(X)$ is a genuine $(h+h' + O(dM))$-pseudopolynomial. Thus we get the equation:
$$ E'(X) \cdot F(X) \cdot G(X) = B(X),$$
where $E'(X)$ is an $(h'+O(dM))$-pseudopolynomial and $B(X)$ is a $(h+h' + O(dM))$-pseudopolynomial, and $\deg(E'), \deg(B) < Mq$.

Applying Lemma~\ref{lem:factormult} to $E'(X)$ and $B(X)$ to analyze factor multiplicities, we get that all irreducible factors of $(F \cdot G)(X)$ have 
factor multiplicity approximately a multiple of $q$, as desired.

\subsection{Proof of Theorem~\ref{thm:pseudoSZ}}

Recall the statement of Theorem~\ref{thm:pseudoSZ}.

\pseudoSZ*

\begin{proof}
First we prove the easier bound:
$$ | \{ \alpha \in \F_q \mid \mult(A, \alpha) \geq M \} | \leq k.$$
Since $A$ is nonzero of degree at $<Mq$, there exists some $\ell < M$ and some $\alpha_0$ with $A_{\langle \ell \rangle}(\alpha_0)=  A^{[\ell]}(\alpha_0)  \neq 0$. For this $\ell$, we have that $A_{\langle \ell \rangle}(X)$ is a nonzero polynomial of degree at most $k$, and thus there are at most $k$ values of $\alpha$ where $A^{[\ell]}(\alpha) = A_{\langle \ell \rangle}(\alpha) = 0$. Thus  $\mult(A, \alpha) > \ell$ for at most $k$ values of $\alpha \in \F_q$, as desired.

Now we prove the other bound.
Write $A(X) = \sum_{i=0}^{c-1} A_i(X) X^{iq}$, where $\deg(A_i) \leq k + M$.

Then for $\ell < M$,
$$A^{[\ell]}(X) = \sum_{i=0}^{c-1} A_i^{[\ell]}(X) X^{iq}.$$

For $\alpha \in \F_q$, this lets us get a simple criterion for when $\mult(A, \alpha) \geq M$:
namely, if for all $\ell < M$,
$$ \sum_{i=0}^{c-1} A_i^{[\ell]}(\alpha) \alpha^i = 0.$$

Let $R(X)$ be the $M \times c$ matrix, with rows indexed by $\ell \in \{0,1, \ldots, M-1\}$, and columns indexed by $i \in \{0,1, \ldots, c-1\}$, whose $(\ell, i)$ entry is
$A_i^{[\ell]}(X)$.

Let $v(X)$ be the vector $(1, X, X^2, \ldots, X^{c-1})$.

Then $\mult(A, \alpha) \geq M$ if and only if:
$$ R(\alpha) \cdot v(\alpha) = 0,$$
namely, if:
$$ v(\alpha) \in \ker(R(\alpha)).$$

We now bound the number of such $\alpha$.
Let $W \subseteq \F_q^c$ be the set of all vectors $(w_0, \ldots, w_{c-1}) \in \F_q$ with $\sum_i w_i A_i(X) = 0$. 
Then for every $w \in W$, 
$$ R(X) \cdot w = 0.$$
Thus for every $\alpha \in \F_q$, 
$$ W \subseteq \ker(R(\alpha)).$$

Our bound for the number of $\alpha \in \F_q$
for which $v(\alpha) \in \ker(R(\alpha))$
then follows from the following two claims.

\begin{claim}
 The number of $\alpha \in \F_q$ for which
 $\ker(R(\alpha)) \neq W$
 is at most $\frac{c}{M-c+1} \cdot k$.
\end{claim}

\begin{claim}
 The number of $\alpha \in \F_q$ for which $v(\alpha) \in W$ is at most $c$.
\end{claim}

The latter claim simply follows from the fact that $\dim(W) < c$, and that the moment curve in $\F_q^c$ does not intersect any proper subspace in more than $c$ points.

It remains to prove the first claim.

Let $I \subseteq \{0,1,\ldots, c-1\}$ be the maximal subset with $\left( A_i(X) \right)_{i \in I}$ linearly independent over $\F_q$. Then $|I| + \dim(W) = c$, and $|I| \geq 1$.

Let $S(X)$ denote the submatrix of $R(X)$ consisting
of those columns indexed by $I$. 

Now suppose $\alpha$ is such that the columns of $S(\alpha)$ are linearly independent. We get that $ \rank(R(\alpha)) = \rank(S(\alpha)) = |I|$, and thus $\dim(\ker(R(\alpha))) =  c - |I| =  \dim(W)$.
This means that $\ker(R(\alpha)) = W$.

To finish the proof, we show that there can be at most $\frac{c}{M-c+1} \cdot k$ values of $\alpha$ for which the columns of $S(\alpha)$ are linearly dependent.
This is simple case of the argument from Theorem 17 of~\cite{GK-subspace-design} (see also~\cite{BCDZ25}). 

Now we use the primality of $q$.
By the Wronskian criterion for linear independence (see~\cite{BD10-wronskian}), the top $|I| \times |I|$ submatrix of $S(X)$ has nonzero determinant $H(X)$.
If $\alpha$ is such that the columns of $S(\alpha)$ are linearly dependent, then
$\alpha$ is a root of multiplicity at least $M-c+1$ of $H(X)$. This is because for each $\ell < M-c+1$, $H^{[\ell]}(\alpha)$ is a linear combination of determinants\footnote{Here we use the product rule for Hasse derivatives, the multilinearity of the determinant, and the formula for iterated Hasse derivatives: $(F^{[i]})^{[j]}(X) = {i + j \choose j} F^{[i+j]}(X)$.} of $|I| \times |I|$ 
submatrices of $S(\alpha)$, which are all $0$.
But $H(X)$ is a nonzero polynomial of degree at most $ck$. 
Thus there can be at most $\frac{c}{M-c+1}\cdot k$ such $\alpha$.

This completes the proof of the first claim, and with that, the theorem.
\end{proof}

\newpage

\bibliographystyle{alpha}  
\bibliography{references}

\appendix
\section{Algebraic characterization of the pseudodegree}
\label{sec:alg-char}

We now prove the algebraic characterization of the pseudodegree.

\algchar*

\begin{proof}
Using the product rule for Hasse derivatives:
 \begin{align*}
 A^{[\ell]}(X) &= \sum_{i \geq 0} \left( A_i \cdot \Lambda^i \right)^{[\ell]}(X)\\
 &= \sum_{i\geq 0}  \left( \sum_{\substack{\ell_0, \ell_1, \ldots, \ell_i\\\sum \ell_j = \ell}} A_i^{[\ell_0]}(X) \cdot \prod_{j=1}^{i}\Lambda^{[\ell_j]} (X) \right)\\
 \end{align*}
 
Reducing these terms mod $\Lambda(X)$,
the only terms that survive have all
$\ell_1, \ldots, \ell_i \in \{1,q\}$,
and thus $\ell \geq i$ and $\ell_0 \in [0, \ell - i]$.

So
\begin{align}
 \nonumber
A^{[\ell]}(X) \mod \Lambda(X) &=  \sum_{i=0}^{\ell}  \left( \sum_{\substack{\ell_1, \ldots, \ell_i \in \{1,q\}\\
\nonumber
\ell_0 = \ell - (\ell_1 + \ldots + \ell_i)}} A_i^{[\ell_0]}(X) \cdot \prod_{j=1}^{i}\Lambda^{[\ell_j]} (X) \right)\\
&= \sum_{i = 0}^{\ell} \sum_{\ell_0 =0}^{\ell-i}  A_i^{[\ell_0]}(X)  \cdot \left( \sum_{\substack{\ell_1, \ldots, \ell_i \in \{1,q\}\\
\nonumber
(\ell_1 + \ldots + \ell_i) = \ell - \ell_0}} \prod_{j=1}^{i}\Lambda^{[\ell_j]} (X) \right)\\
\nonumber
&= \sum_{i =0}^{\ell} \sum_{\ell_0 =0}^{\ell-i}  A_i^{[\ell_0]}(X)  \cdot  C_{i, \ell, \ell_0,q} \quad \mbox{for some integer $C_{i, \ell_0, \ell, q}$} \\
\label{eq:pder-formula}
&= (-1)^\ell A_\ell(X)  + \sum_{i=0}^{\ell-1} \sum_{\ell_0 =0}^{\ell-i} A_i^{[\ell_0]}(X) \cdot C_{i, \ell, \ell_0, q} \quad \mbox{isolating the $i = \ell, \ell_0 = 0$ term}
\end{align}

Thus if all the $A_i$ have degree at most $k$, 
we get that $A^{[\ell]}\mod \Lambda(X) = A_{\langle \ell \rangle}(X)$ has degree at most $k$.

In the reverse direction, the above formula lets us express $A_\ell(X)$ as $\pm A_{\langle \ell \rangle}(X) + $ a linear combination of $A_{i}^{[j]}$,
with $i < \ell$. By induction on $\ell$, we get that $\deg(A_\ell)$ has degree at most
$\max_{i \leq \ell} \deg(A_{\langle i \rangle})$.
 
\end{proof}

\section{$m$-th power residue character}

Let $q$ be a prime power.
Let $m$ be a prime dividing $(q-1)$.
Let $\Gamma_m$ be the set of $m$'th roots of $1$ in $\F_q$.
Let $\chi_m : \F_q \to \{0\} \cup \Gamma_m$ be the ($\F_q$-valued) $m$-th power residue character
given by: $\chi_m(\alpha) = \alpha^{(q-1)/m}$.

Let $g(X)$ be monic of degree at most $d \leq O_m(\epsilon \sqrt{q})$, with each irreducible factor appearing with multiplicity $\in \{1,2, \ldots, m-1\}$.

The following algorithm can efficiently find $g(X)$ given an $r : \F_q \to \{0\} \cup \Gamma_m$ with $\Delta(\chi \circ g, r)$ being at most a small enough constant fraction of $q$. As in the case of $m=2$, this decoding radius falls short of the unique decoding radius, which is about $\frac{1}{2}\left( 1 - \frac{1}{m} \right)q$.

{\noindent \bf Algorithm A$_m$:}\\
{\noindent Parameters: degree $d \leq O(\frac{\epsilon}{m}\sqrt{q})$, error-bound $e \leq \left( \frac{1}{12} - \epsilon \right) q$.}\\
{\noindent Input $r: \F_q \to \{0\} \cup \Gamma_m$ }
\begin{enumerate}
 \item Set
 \begin{itemize}
  \item $M = \frac{16}{\epsilon}dm$
  \item $c = \frac{M}{2}$
  \item $h = 2 \cdot e$.
  \item $D = d \cdot  \left( (q-1) \cdot \frac{m-1}{m} + M \right) + cq = (1 + O(\epsilon)) \cdot \frac{1}{2} \cdot M \cdot q$,
  \item $u = h + (m-1)dM = 2e + O(\frac{m d^2}{\epsilon} )$.
 \end{itemize}
 
 \item For each $\nu \in \{1, 2, \ldots, m-1\}$:
 \begin{itemize}
 \item
 Solve an $\F_q$ system of linear equations to find polynomials
 $F_{\nu}(X), U_{\nu,0}(X), \ldots, U_{\nu,M-1}(X) \in \F_q[X]$, not all zero,
 with:
 \begin{itemize}
 \item $ \deg(F_\nu) \leq D$,
 \item for each $\ell$, $\deg(U_{\nu,\ell}) \leq u$,
 \item For all $\alpha \in \F_q$, $0 \leq \ell < M$:
 \begin{align}
 \label{eq:linsyst-m}
  F_{\nu}^{[\ell]}(\alpha)  = r^{\nu}(\alpha) \cdot U_{\nu,\ell}(\alpha).
 \end{align}

 \end{itemize}
 \end{itemize}
 \item Factor $F_{\nu}(X)$ into irreducible factors:
 $$ F_{\nu}(X) = \lambda_{\nu} \prod_j  H_j^{\mu_{\nu,j}}(X),$$
 where the $H_j(X) \in \F_q[X]$ are distinct and monic irreducible polynomials, and $\lambda \in \F_q^*$.
 \item For each $j$, define $\mu_j \in \{0,1,\ldots, m-1\}$ to be the unique number such that for all $\nu \in \{1, \ldots, m-1\}$, we have:
 $$ \mu_{\nu, j} - \frac{\nu \cdot \mu_j }{m} \cdot (q-1) \in   \left(-\left(\frac{1}{12} - \epsilon\right)q, \quad \left(\frac{3}{12} - \epsilon \right) q  \right) \mod q.$$
 \label{getJstep-m}
 \item Set
 $$f(X) = \prod_{j} H_j^{\mu_j}(X).$$
 \item Return $f(X)$.
\end{enumerate}

The analysis is very similar to the analysis in Section~\ref{sec:chi}.

Let $G(X) = g^{(q-1)/m + M}(X)$, which is a $(\chi_m \circ g)$-twisted $(O(dM), M)$-pseudopolynomial.

To show that the linear equations have a solution, we first interpolate a high multiplicity error-locator pseudopolynomial $E(X)$ with $\deg(E) < cq$ and $\pdeg(E) < h$. Then we show that taking $F_\nu = E \cdot G^{\nu}$ yields a solution to Equations~\eqref{eq:linsyst-m}.

Next we show that whatever $F_\nu$ the algorithm finds is related to $G$.
Exactly as in Section~\ref{sec:chi}, we find polynomials $A_{\nu}(X), B_\nu(X)$
with $\deg(A_\nu) < \frac{3}{8} Mq$, $\deg(B_\nu) < \frac{7}{8} Mq$,
and $\pdeg(A_\nu) \leq e$, $\pdeg(B_\nu) \leq 3e + O(dM)$ such that:
$$A_\nu(X) \cdot F_\nu(X) = B_\nu(X) \cdot G^\nu(X).$$

From these equations, we will be able to extract all the irreducible factors of $g(X)$.
Let $H(X)$ be an irreducible polynomial, and let $\mu_{A_\nu}, \mu_{F_\nu}, \mu_{B_\nu}, \mu_{G}, \mu_g$ be its factor multiplicity in $A_\nu, F_\nu, B_\nu, G, g$ respectively.
Then the above identity gives us:
$$\mu_{A_\nu} + \mu_{F_\nu} = \mu_{B_{\nu}} + \nu \cdot \mu_{g}  \cdot \frac{q-1}{m}.$$
Thus for all $\nu$:
$$\mu_{F_\nu} - \nu \cdot \mu_g \cdot \frac{q-1}{m}  = \mu_{B_{\nu}} - \mu_{A_{\nu}} \in q \cdot \Z + \left[-e ,  3e + O(dM)  \right] \subseteq q\cdot \Z + \left( -\left(\frac{1}{12} - \epsilon\right) q, \left(\frac{3}{12} - \epsilon\right)q\right).$$
Thus $\mu_g$ is a solution found in Step~\ref{getJstep-m} of the algorithm.
Why is there no other solution? Suppose $\mu' \in \{0,1,\ldots, m-1\}$ is distinct from $\mu_g$ and satisfies:
$$\mu_{F_\nu} - \nu \cdot \mu' \cdot \frac{q-1}{m} \in  q\cdot \Z +  \left( -\left(\frac{1}{12} - \epsilon\right) q, \;\left(\frac{3}{12} - \epsilon\right)q\right)$$
for all $\nu$.
Then subtracting the two equations and setting $\mu'' = \mu' - \mu_g$, we get that for all $\nu \in \{1, \ldots, m-1\}$:
$$ \nu \cdot \mu'' \cdot \frac{q-1}{m} \in q \cdot \Z + \left(-\left(\frac{1}{3} - 2\epsilon\right) q, \; \left(\frac{1}{3} - 2\epsilon\right) q\right).$$
This is impossible: for any $\mu'' \in \{1,\ldots, m-1\}$, there is a $\nu \in \{1, \ldots, m-1\}$ such that $\frac{\mu'' \cdot \nu}{m} \in \Z + \left[\frac{1}{3}, \frac{2}{3}  \right] $.

\section{Additive characters over fields of small characteristic $p$}

We only give the analogue of the main subroutine, Algorithm B, from Section~\ref{sec:tr}.

Let $q = p^b$ with $p$ prime.
Let $\Tr: \F_q \to \F_p$ be the field trace map.

Let $g(X)$ be a polynomial of degree at most $d \leq O_m(\epsilon \sqrt{q})$, with only monomials of degree relatively prime to $m$.

The following algorithm can efficiently find the leading coefficient of $g(X)$ given an $r : \F_q \to \F_p$ with $\Delta(\Tr \circ g, r)$ being at most a small enough constant fraction of $\frac{1}{p} q$. For $p = O(1)$ this is a constant fraction of $q$, but for growing $p$ it is far from satisfactory. The distance of the underlying code is about $(1- \frac{1}{p}) q$ (it actually grows with $p$), and there ought to be an efficient unique decoding algorithm to radius half of that.

{\noindent \bf Algorithm B$_p$:}\\
{\noindent Parameters: degree $d \leq O_p(\epsilon \sqrt{q})$ (an integer not divisible by $p$), error-bound $e \leq \left( \frac{1}{4 p}  - \epsilon \right) q$.}\\
{\noindent Input $r: \F_q \to \F_p$ }
\begin{enumerate}
  
 \item Set
 \begin{itemize}
  \item $M = \frac{16}{\epsilon}d$
  \item $c = \frac{1}{2} \cdot M$
  \item $h = 2 \cdot e$
  \item $u = h + dM = 2e + O(\frac{ d^2}{\epsilon} )$.
 \end{itemize}

 \item\label{Blinsyst1-p} For each $h^* \in [0,h]$
 \begin{itemize}
 \item Try to solve a system of $\F_q$-linear equations to find polynomials
 $E(X)$, $F(X)$,\\
 $(U_{\ell}(X))_{\ell \in [0,M-1]}$,
 where:
 \begin{itemize}
  \item $F(X)$ has degree exactly $(d/p + c)q + h^*$.
  \item $E(X)$ has degree exactly $cq + h^*$, and is of the form:
  $$ \sum_{i=0}^{c} E_i(X) \Lambda^{i}(X),$$
  where $\deg(E_i) \leq h$.
  \item each $U_{\ell}(X)$ has degree at most $u$.
 
 \item For all $\alpha \in \F_q$, $0 \leq \ell < M$:
 \begin{align}
 \label{eq:Blinsyst1-p}
  F^{[\ell]}(\alpha)  = r(\alpha) \cdot E^{[\ell]}(\alpha) +  U_{\ell}(\alpha),
 \end{align}

 \end{itemize}
 \end{itemize}
 
 \item If no such $h^*$ exists, return $0 \cdot X^d$.
 \item Otherwise take one solution $E(X), F(X)$,
 with degrees $D_E, D_{F}$ respectively.
 \item Let $a\cdot X^{D_{F}}$, $b \cdot X^{D_E}$ be the leading monomials of $F(X), E(X)$ respectively.
 
 Set $a_d = (a/b)^p$.
 
 \item Return $a_d X^d$.
 \end{enumerate}

The proof of correctness is essentially identical to the proof of Theorem~\ref{thm:B} in Section~\ref{sec:tr}, and is omitted.
 
\end{document}